\documentclass[11pt]{llncs}
\usepackage{fullpage}

\usepackage{amsmath, amssymb}

\usepackage{amsthm}
\usepackage[margin=1in]{geometry}
\usepackage[T1]{fontenc}
\usepackage[usenames, dvipsnames]{color}
\usepackage[numbers]{natbib}
\usepackage{algorithmic}
\usepackage{tikz}
\usepackage{enumerate}
\usepackage{hyperref}
\usepackage{graphicx}
\usepackage{caption}
\usepackage{subcaption}
\usepackage{cleveref}
\usepackage{enumitem,linegoal}
\usepackage[linesnumbered,ruled,vlined]{algorithm2e}
\usepackage{comment}	
\usepackage{mathtools}
\usepackage{multirow}
\usepackage{dsfont}
\usepackage{thm-restate}
\usepackage[dvipsnames]{xcolor}
\usepackage{tcolorbox}
\usepackage{setspace}

\newcommand{\opt}{\mathsf{opt}}
\newcommand{\pref}{\sigma}
\newcommand{\prefp}{\vec \sigma}

\newcommand{\poi}{\mathsf{PoII}}
\newcommand{\cns}{\rhd^{\!v}_{\!\alpha}}
\newcommand{\scns}{\rhd_{\!\alpha}}
\newcommand{\ggc}{{\mathrel{{\succ}\!\!{\succ}}}}

\newcommand{\candid}{A}
\newcommand{\PP}{$\mathcal{P}$\xspace}
\newcommand{\QQ}{$\mathcal{D}$\xspace}

\DeclareMathSymbol{\mlq}{\mathord}{operators}{``}
\DeclareMathSymbol{\mrq}{\mathord}{operators}{`'}

\newcommand{\qsucc}{\mlq\!\succ\!\mrq}
\newcommand{\qggc}{\mlq\!\ggc\!\mrq}
\newcommand{\intns}{\Join}

\renewcommand{\leq}{\leqslant}
\renewcommand{\ge}{\geqslant}
\renewcommand{\geq}{\geqslant}

\newcommand{\ags}{\ensuremath{N}}
\newcommand{\alts}{\ensuremath{A}}
\newcommand{\elc}{\mathcal{E}}

\newcommand{\scs}{\ensuremath{\mathsf{sc}}}
\newcommand{\scsd}[1]{\ensuremath{\mathsf{sc}_{#1}}}
\newcommand{\rank}[1]{\ensuremath{\mathsf{rank}_{#1}}}
\newcommand{\sco}{\scs}

\DeclareMathOperator{\psm}{PSM}

\newcommand{\dist}{\mathsf{dist}}
\newcommand{\distv}{\mathsf{dist}^{v}_\alpha}

\DeclareMathOperator*{\minimize}{minimize}
\DeclareMathOperator*{\argmin}{arg\,min}

\newcommand{\plu}{\mathsf{Plu}}
\newcommand{\optaw}{\opt^{\mathsf{aw}}_{\alpha}}
\newcommand{\optob}{\opt^{\mathsf{ob}}_{\alpha}}

\newcommand{\optawv}{\opt^{v-\mathsf{aw}}_{\alpha}}
\newcommand{\optobv}{\opt^{v-\mathsf{ob}}_{\alpha}}
\newcommand{\poiv}{\poi^{v}}

\newcommand{\butterfly}[1]{
 \sigma
}

\allowdisplaybreaks[1]

\begin{document}

\title{Metric Distortion with Preference Intensities}

\author{Mehrad Abbaszadeh\inst{1}\and Ali Ansarifar\inst{2}\and Mohamad Latifian\inst{3} \and Masoud Seddighin\inst{4}}

\institute{
Georgia Institute of Technology, United States
\and
Sharif University of Technology, Iran\and University of Edinburgh, United Kingdom\and Tehran Institute for Advanced Studies, Iran}

\maketitle

\begin{abstract}

\normalsize In voting with ranked ballots, each agent submits a strict ranking of the form $a \succ b \succ c \succ d$ over the alternatives, and the voting rule decides on the winner based on these rankings. Although this ballot format has desirable characteristics, there is a question of whether it is expressive enough for the agents. \citet{KLS23} address this issue by adding intensities to the rankings. They introduce the \emph{ranking with intensities} ballot format, where agents can use both $\qggc$ and $\qsucc$ in their rankings to express intensive and normal preferences between consecutive alternatives in their rankings.
While they focus on analyzing this ballot format in the utilitarian distortion framework, in this work, we look at the potential of using this ballot format from the metric distortion viewpoint. We design a class of voting rules coined \emph{Positional Scoring Matching} rules, which can be used for different problems in the metric setting, and show that by solving a zero-sum game, we can find the optimal member of this class for our problem. This rule takes intensities into account and achieves a distortion lower than $3$. In addition, by proving a bound on the price of ignoring intensities, we show that we might lose a great deal in terms of distortion by not taking the intensities into account.

\end{abstract}

\clearpage
\section{Introduction}
The field of computational social choice focuses on aggregating the individual opinions of multiple agents to arrive at a collective decision that affects all of them. This process typically involves two phases: the \emph{elicitation} of the agents' preferences and the \emph{aggregation} of these preferences to determine the final output. One key element of the elicitation phase is the \emph{ballot} format that is used to capture the agents' preferences. The ballot design is crucial, as it dictates how agents report their preferences, and this inherently constrains the collected information, which, in turn, can influence the voting outcome and potentially lead to decisions that do not accurately reflect agents' true preferences. For example, ranked ballots and approvals are two of the most famous ballot formats that are typically used in voting mechanisms. Although these ballot formats are intuitive and capture varying details about agents' desired outcomes, it remains unclear which format is most appropriate. In different collective decision-making settings, agents may be able to provide varying types of information, and alternative elicitation methods may be more suitable. Thus, the choice of ballot format remains a context-dependent question, and it is essential to evaluate whether standard formats accurately express agents' preferences.

In this work, we focus on single-winner voting and a practical modification of ranked ballots that captures more detail about agents' preferences. In standard ranked ballots, each agent submits a complete ordering of the alternatives, and the aggregation rule selects a winner based on these rankings. However, in many real-world settings, agents not only rank alternatives but also differ in how strongly they prefer them. They may strongly favor some options, feel indifferent about others, or actively oppose certain choices. Simple ordinal rankings obscure these distinctions by treating all adjacent positions as equally spaced. For example, in internship selection at a tech firm, team leaders often classify candidates into broad tiers like ``top-tier'', ``acceptable'', and ``poor fit''. These categories reflect preference gaps that simple rankings fail to capture.
 Ignoring such information can lead to undesirable outcomes.  Similarly, in R\&D evaluations, committee members often use rating scales (e.g., 0--100) to rank proposals. While these scores convey richer information than simple rankings, eliciting precise numerical values is often impractical.
 In response, several methods have been proposed to capture and elicit the intensity of preferences ~\cite{cook1985ordinal,harvey1999aggregation,saaty2007group,farquhar1989preference}.

 A natural and implementable idea is to augment ranked ballots with a notion of \emph{preference intensities}, allowing voters to express not just order but strength of preference. While a traditional ranked ballot requires a strict ranking of alternatives, it cannot differentiate between slight and strong preferences. Consider a group choosing between options of what to eat, with their options being Pizza (\(p\)), Steak (\(s\)), or Burger (\(b\)). Alice may submit \(p \succ s \succ b\) reflecting a mild preference for Pizza. Bob, a vegetarian, might submit the same ranking even though Pizza is his only viable option, and his preference is far stronger. This lack of distinction in preference intensities can lead to sub-optimal outcomes. To address this limitation, \citet{KLS23} introduced the \emph{ranking with preference intensities} ballot format, which extends traditional rankings by permitting agents to use \(\qggc\) or \(\qsucc\) to indicate stronger or weaker preferences. For example, Bob could express his preferences more accurately by submitting \(p \;\ggc b \succ s\), and a voting rule can then leverage this additional information to make a more efficient decision.

There are various approaches in the literature to compare voting rules. Some studies compare them based on axiomatic properties, but evaluating decision efficiency requires a quantitative measure.
One approach that has recently received significant attention is the notion of \emph{distortion}, introduced by \citet{procaccia2006distortion}.
 Distortion measures the efficiency loss in a voting process by comparing the outcome based on reported preferences to the ideal outcome under full information.
  Distortion has been studied in two main frameworks: utilitarian and metric.

In the utilitarian framework, introduced by \citet{procaccia2006distortion}, it is assumed that each agent has a subjective unit sum utility function over the alternatives and reports her ranking accordingly. The goal is to use the reported information to select an alternative that maximizes the social welfare. In this context, distortion of a rule is defined as the worst-case ratio between the maximum achievable social welfare and the social welfare of the alternative chosen by the rule. \citet{KLS23} analyze the utilitarian distortion of some rules that use ranking with preference intensities ballot format, and show how ignoring the intensities can affect the utilitarian distortion by introducing the notion of \emph{price of ignoring intensities}.

The metric distortion framework, first studied by \citet{anshelevich2018}, assumes that agents and alternatives are embedded in a metric space, and agents submit preferences based on proximity (closer alternatives are preferred). One motivation for this is the facility location problem~\cite{anshelevich2021ordinal,filos2021approximate}, where the goal is to select facility locations to minimize total distance or cost of serving clients. In this context, distortion is defined as the worst-case ratio of the social cost of the chosen alternative to the minimum possible social cost. \citet{anshelevich2018} proved a lower bound of 3 on the metric distortion of any deterministic voting rule. They also proved that the metric distortion of the Copeland rule is 5. Eventually, \citet{GHS20} closed the gap by introducing the \emph{Plurality Matching} rule, which achieves a distortion of 3.\\ 
Both metric and utilitarian distortion have been studied for various ballot formats. We refer to the survey of \citet{anshelevich2021distortion} for an extensive overview of these results. In this work, we study the limits of metric distortion achievable with the ranking-with-preference-intensities ballot format and compare them to those of rules using simple ranked ballots.
 We draw a complete picture by analyzing distortion in the general case, in the line metric, and by establishing bounds on the price of ignoring intensities.

\subsection{Related Work}

Distortion was introduced by \citet{procaccia2006distortion} and has since received extensive attention; see \citet{anshelevich2021distortion} for a comprehensive review. Here we overview the works most relevant to ours.

In general, the problem we consider is single-winner voting with ranked preferences. In this setting, the optimal achievable distortion in the utilitarian setting is $\Theta(m^2)$, with $m$ being the number of alternatives, and the plurality rule guarantees this bound \cite{caragiannis2017subset,caragiannis2011voting}. In the metric setting, the distortion of any deterministic single-winner voting rule with ranked preferences has a lower bound of 3. Over time, subsequent studies have progressively reduced the upper bound on the optimal metric distortion, from 5 \cite{anshelevich2018} to 4.236 \cite{Munagala2019}, and ultimately to 3 \cite{GHS20,Kizilkaya2022}. 
   In addition, \citet{gkatzelis2023best}  suggested  a voting rule that  achieves asymptotically optimal utilitarian and metric distortion simultaneously. 

Some studies suggest that having extra information can significantly reduce the distortion. This additional information may come from assumptions about the context of the election, agents' behavior, or ballot format.
For instance, in two successive studies, \citet{cheng2017people,cheng2018distortion} examine cases where alternatives are drawn randomly from the agent population, which resulted in improved distortion bounds for various voting rules. \citet{seddighin2021distortion} considered a scenario where the agents can abstain based on a  behavioral model. They show that for two alternatives, this model reduces the metric distortion.
\citet{anshelevich2017randomized} introduce $\alpha$-decisive metrics, where each agent's distance to their closest alternative is at most $\alpha$ times the distance to their next closest alternative. $\alpha$-decisive metrics have been considered in many follow-up studies~\cite{Gross2017,GHS20}. Our framework is closely related to this setting, as it allows agents to specify whether they are $\alpha$-decisive between any two consecutive alternatives.

\citet{amanatidis2022few} studied deterministic rules in the utilitarian setting, allowing a limited number of value queries in addition to agents’ preference lists. They showed that one query per agent yields distortion \( O(m) \), \( O(\log m) \) queries give \( O(\sqrt{m}) \) distortion, and \( O(\log^2 m) \) queries suffice for constant distortion.
More broadly, \citet{mandal2019efficient} and \citet{mandal2020optimal} concentrated on the communication complexity of voting algorithms. They established a framework that collects limited information from agents, without prior knowledge of their preferences, and optimizes the trade-off between distortion and the quantity of information obtained from each agent.

Finally, distortion has been extended to more general scenarios, including committee selection \cite{caragiannis2022metric,BNPS21,bedaywi2023distortion}, participatory budgeting \cite{benade2021preference,flanigan2023distortion}, and different ballot formats \cite{abramowitz2019awareness, amanatidis2021peeking, ma2020improving, bishop2022sequential, BHLS22}.

\subsection{Our contribution} 

In this work, we assume that agents and alternatives are embedded in a metric space with distance function \( d \), and we use a parameter \( \alpha \in [0, 1] \) to model the intensity of preference. Specifically, an agent strongly prefers alternative \( a \) over \( b \), denoted \( a \ggc b \), if their distance to \( a \) is less than \( \alpha \) times their distance to \( b \). We analyze our methods as a function of \( \alpha \).

Metric distortion with preference intensities can be analyzed in two settings: \emph{mandatory} and \emph{voluntary} elicitation of the intensities. In the former, agents use $\qggc$ if and only if there is a gap in their distances, while in the latter, agents might still use  $\qsucc$ when there is a gap in their distances.

It is evident that the optimal distortion in the voluntary setting is 3 (matching the bound on the distortion with ranked ballot), since agents might decide not to report any intensities. However, the distortion bound in the mandatory setting is unknown. We start by proving lower bounds on the metric distortion. Our main technical contribution is introducing a deterministic rule and analyzing its distortion in this setting.
We start by analyzing the class of \emph{moderate-up-to-$k$} preferences, which is the class of preference profiles in which each agent reports his first intense preference in the $k\textsuperscript{\text{th}}$ position. To this end, we generalize the Plurality Matching rule to the class of rules coined \emph{Positional Scoring Matching rules}. We analyze the distortion of members of this class and prove that finding the optimal rule among them is equivalent to finding the value of a two-player zero-sum game. We believe that this class of rules can be useful in different contexts in the metric setting and also in the other settings, and the framework to analyze them might be of interest for future work.

The next step is to use this analysis to design a rule for the general case, with the distortion depending on $\ell_{\max}$, the maximum $\ell \in [m - 1]$ for which all the agents are moderate-up-to-$\ell$. We prove that the distortion of this rule is much less than $3$ when $\ell_{\max}$ is small
(see \Cref{fig:mand_k_1,fig:mand_k_2}). We further show that this bound is robust by proving that  if a small portion of the agents are not 
moderate-up-to-$k$ for any $k \leq \ell_{\max}$, our bound does not change drastically.

To give more meaning to our results, it is worth mentioning that there are reasons to believe that in real world elections most of the agents are moderate-up-to-$k$ for a small $k$: agents tend to partition the alternatives into the small set of desirable outcomes and the large set of undesirable ones. Approval voting is motivated by this principle. Moreover, agents allocate more cognitive and emotional resources to learn more about their high-ranked options, where the stakes are perceived to be higher, and hence they usually have a strong preference for their top few options over the rest.

In \Cref{sec:line}, we consider the special case where our metric is a line. We establish a lower bound of \(\max\left( \frac{3 - \alpha}{1 + \alpha}, 2\alpha + 1 \right)\) on the distortion in the mandatory setting for instances with two alternatives. We also propose a rule for two-alternative cases, which we conjecture achieves a similar distortion to this lower bound. Furthermore,
we provide another lower bound, \(\frac{1 + 3\alpha^{-\lfloor\frac{m}{2}\rfloor}}{3 + \alpha^{-\lfloor\frac{m}{2}\rfloor}}\), showing that distortion in the line metric can still approach values close to 3, similar to the unrestricted case.

In \Cref{sec:poi}, we study the \emph{Price of Ignoring Intensities (POII)} which measures the efficiency loss that accrues if simple ranked ballots are used when intensities are available. We show that in the mandatory setting POII is at least $$
\frac{3 \left(\alpha^{\lfloor m \rfloor_{\text{even}}}+1\right)}{1+2\alpha^{\lfloor \frac{m}{2} \rfloor}-\alpha^{\lfloor m \rfloor_{\text{even}}}}\cdot
$$
In this section, we also prove the lower bound $\frac{3}{2 \alpha^{\lfloor\frac{m}{2}\rfloor} +1}$ on POII in the voluntary setting. 
 To prove these bounds, we carefully analyze different possible outcomes on a specific instance. This requires formulating our problem as an LP and performing primal-dual analysis. We modify the techniques introduced by \citet{kempe2020analysis} to fit our setting.

\section{Preliminaries}

For \( t \in \mathbb{N} \), let \([t] = \{1, \ldots, t\}\), and for set $S$, let $\Delta(S)$ be the set of distributions over $S$. An election $\elc = (\ags, \alts, \prefp)$ involves a set of \(n\) agents \(\ags\), a set of \(m\) alternatives \(\alts\), and an elicited preference profile \(\prefp\). We use \(i, j\) to refer to agents and \(a, b, c\) to refer to alternatives.

\paragraph{Preference Elicitation.}
In this work, we focus on preference profiles elicited through a ranking with intensities ballot format. In this ballot format, each agent \(i\) submits a preference \(\pref_i = (\pi_i, \intns_i)\), where \(\pi_i: [m] \to \alts\) is a one-to-one function representing the agent's ranking of the alternatives, and \(\intns_i: [m-1] \to \{\succ, \ggc\}\) indicates the intensity of preferences. Specifically, \(\intns_i(j) = \qsucc\) means agent \(i\) prefers \(\pi_i(j)\) over \(\pi_i(j+1)\), while \(\intns_i(j) = \qggc\) means agent $i$ strongly prefers \(\pi_i(j)\) over \(\pi_i(j+1)\).

Let \(S(\alts)\) denote the set of all possible rankings with intensities over \(\alts\), and \(\prefp = (\pref_1, \ldots, \pref_n)\) represent the preference profile of all agents. Additionally, define \(\rank{i}(a)\) as the rank of alternative \(a\) in the ranking of agent \(i\), and \(\plu(a)\) as the number of agents who rank \(a\) as their top alternative.

\paragraph{Voting Rule.}
A single-winner voting rule \( f \) using the ranking with intensities ballot format is a function \( f: S(\alts)^n \to \alts \) that takes a preference profile as input and selects a winning alternative. We also consider voting rules that only access rankings. In this case, the ranking profile is defined as \( \vec{\pi} = (\pi_1, \ldots, \pi_n) \), and the voting rule uses this profile to determine the winner.

\paragraph{Metric Framework.}
We assume that all agents and alternatives are embedded in a metric space characterized by a distance function \( d: (\ags \cup \alts)^{2} \to \mathbb{R}_{\geq 0} \). The distance function is positive, symmetric, satisfies the triangle inequality, and for any $x\in \alts \cup \ags, d(x,x)=0$.
The social cost of alternative $a \in \alts$ with respect to metric $d$ is defined as $\scs_d(a) := \sum_{i \in \ags} d(i, a)$, and the optimal alternative is $\opt_d := \argmin_{a \in \alts} \scs_d(a)$. If $d$ is clear from the context, we drop it from the subscript. A ranking profile \( \vec{\pi} \) is said to be consistent with metric \( d \), denoted by \( \vec{\pi} \rhd d \), if for each agent \( i \in \ags \) and alternatives \( a, b \in \alts \), \( a \succ_i b \) implies \( d(i, a) \leq d(i, b) \). This means agents rank alternatives in increasing order of their distance from them. 

\paragraph{Interpretation of the Intensities.} Intensities that are reported by the agents can be interpreted in two different ways: \emph{mandatory elicitation} and \emph{voluntary elicitation}. For a fixed \( \alpha \in [0, 1] \), a preference profile \( \prefp \) with intensities is \(\alpha\)-consistent with the metric \( d \) under mandatory elicitation, denoted by \( \prefp \scns d \), if for any agent \( i \in \ags \) and \( j \in [m-1] \), the following holds: if \( \intns_i(j) = \qsucc \), then \( d(i, \pi_i(j+1)) \geq d(i, \pi_i(j)) > \alpha \cdot d(i, \pi_i(j+1)) \); and if \( \intns_i(j) = \qggc \), then \( d(i, \pi_i(j)) \leq \alpha \cdot d(i, \pi_i(j+1)) \). In other words, agent \( i \) expresses a strong preference in the \( j\textsuperscript{\text{th}} \) position \emph{if and only if} there is an \(\alpha\)-gap between the distances to their \( j\textsuperscript{\text{th}} \) and \( j+1\textsuperscript{\text{th}} \) alternatives.

We primarily focus on the mandatory setting, but we also consider voluntary elicitation of the intensities in certain parts of this paper. An intensive preference profile \( \prefp \) is said to be \(\alpha\)-consistent with the metric \( d \) under voluntary elicitation of the intensities, and it is denoted by $\prefp \cns d$, if \( \vec{\pi} \rhd d \) and, for any agent \( i \in \ags \) and any \( j \in [m-1] \), where \( \intns_i(j) = \qggc \), it holds that \( d(i, \pi_i(j)) \leq \alpha \cdot d(i, \pi_i(j+1)) \). In other words, if agent \( i \) expresses a strong preference in the \( j\textsuperscript{\text{th}} \) position, there is a multiplicative \(\alpha\)-gap between the distances to the alternatives ranked \( j\textsuperscript{\text{th}} \) and \( j+1\textsuperscript{\text{th}} \). However, if no intensity is reported, the presence of a multiplicative \(\alpha\)-gap remains uncertain.

\paragraph{Metric Distortion.}
Distortion serves as a quantitative measure of the efficiency of voting rules. For \( \alpha \in [0,1] \), the \(\alpha\)-distortion of an alternative \( a \) in an election \( \elc = (\ags, \alts, \prefp) \) is defined as the worst-case ratio between the social cost of \( a \) and the social cost of the optimal alternative, where the worst case is taken across all metrics consistent with \( \prefp \), i.e.,  $$\dist_\alpha(a, \elc) := \sup_{d: \prefp \scns d} \frac{\scs_d(a)}{\min_{b \in \alts} \scs_d(b)}\cdot$$
We sometimes use the notation $\dist_\alpha(a, \prefp)$ whenever $\ags, \alts$ are clear in the context.\\ 
The metric distortion of a rule \( f \) is then defined as the worst-case distortion of its output across all possible election instances, given by \( \dist_\alpha(f) := \sup_{\elc} \dist_\alpha(f(\prefp), \elc) \). Note that this definition assumes mandatory elicitation of the intensities. A similar definition applies under voluntary elicitation, referred to as voluntary distortion (\( \dist^v_\alpha \)), which is formally defined in \Cref{sec:poi}.

\section{Bounds on the Distortion}
\label{sec:dist}

In this section, we explore the potential to achieve a distortion lower than 3 under the assumption of mandatory elicitation of the intensities. In this context, when agents do not use the \( \qggc \) sign in their intensity reports, it indicates that the distances between consecutive alternatives in their rankings differ by no more than a factor of \( \alpha \). This additional information, beyond what is provided by simple ranked ballots, offers the possibility of surpassing the distortion threshold of 3 using a deterministic rule.
We start the section by proving lower bounds on the distortion. Then we show an upper bound that depends on the location of the first $\qggc$ in the preferences of the agents. Note that by distortion in this section, we mean metric distortion assuming mandatory elicitation of the intensities. 

\subsection{Lower Bound \label{sec:mand_lb}}

We begin by establishing a lower bound on the distortion of deterministic rules using the ranked preferences with intensities ballot format, which is derived from two complementary bounds that together form the overall lower bound.

\begin{theorem}
    \label{thm:mand_lb}
    The metric distortion of any voting rule $f$ using ranked preferences with intensities ballot format with parameter $\alpha \in [0, 1]$ and assuming mandatory elicitation of the intensities is at least
    \noindent
    $1 + 2 \max\left(\alpha, \frac{1 - \alpha^{\lfloor\frac{m}{2}\rfloor}}{1 + \alpha^{\lfloor\frac{m}{2}\rfloor}} \right)\cdot$
\end{theorem}

As mentioned, to prove this lower bound, we present two different instances. The lower bound from \Cref{lem:mand_lb_left} works better for smaller values of \( \alpha \), while the bound from \Cref{lem:mand_lb_right} becomes tighter as \( \alpha \) approaches 1. \Cref{fig:mand_lb2} illustrates how these two bounds look like for instances with $m \in \{5, 10, 15, 25\}$. 



\begin{figure}[t]
\centering
\begin{minipage}{.48\textwidth}
  \includegraphics[width=\textwidth]{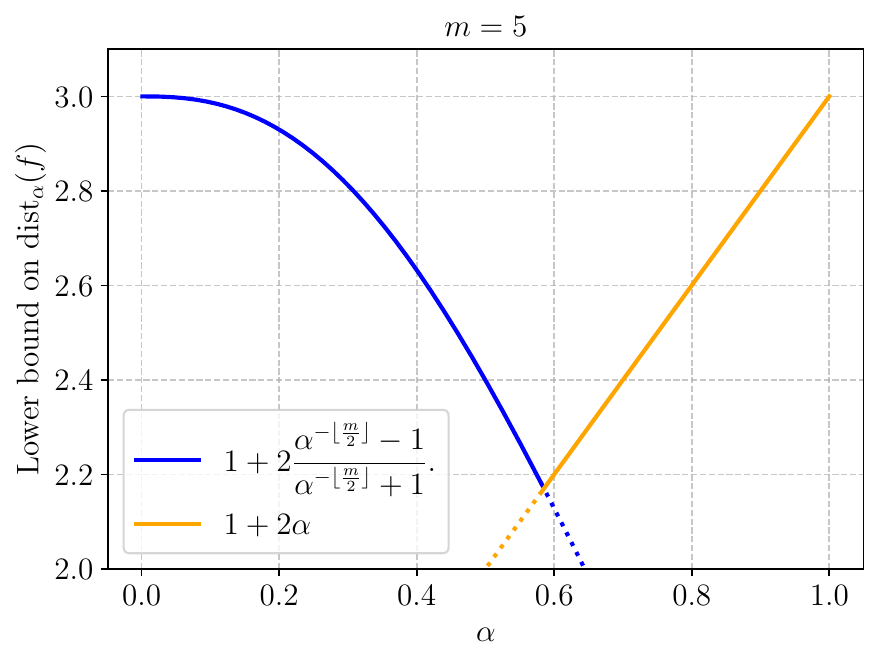}
  \end{minipage}
  \hfill
  \begin{minipage}{.48\textwidth}
  \includegraphics[width=\textwidth]{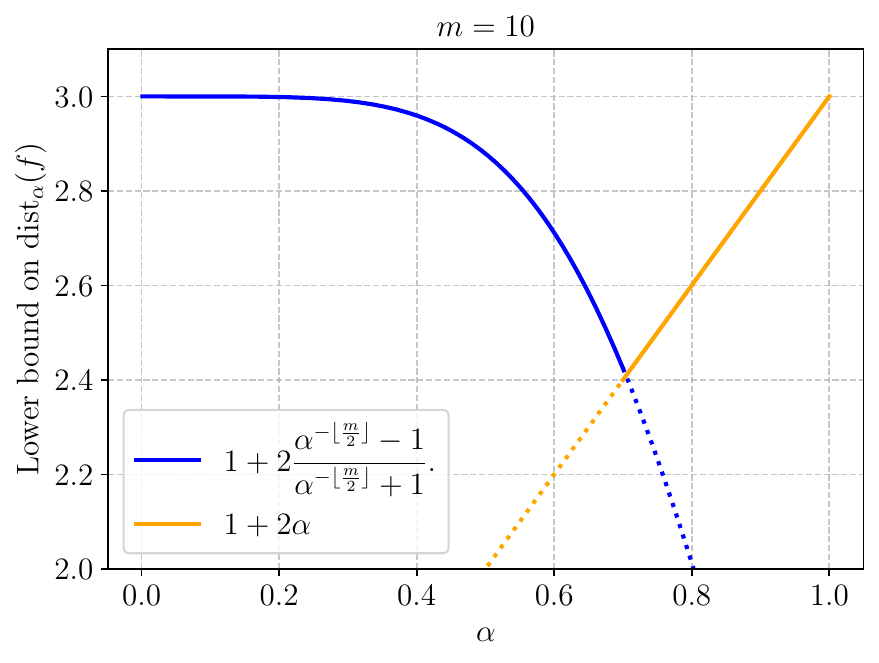}
  \end{minipage}
  \begin{minipage}{0.48\textwidth}
  \centering
  \includegraphics[width=\textwidth]{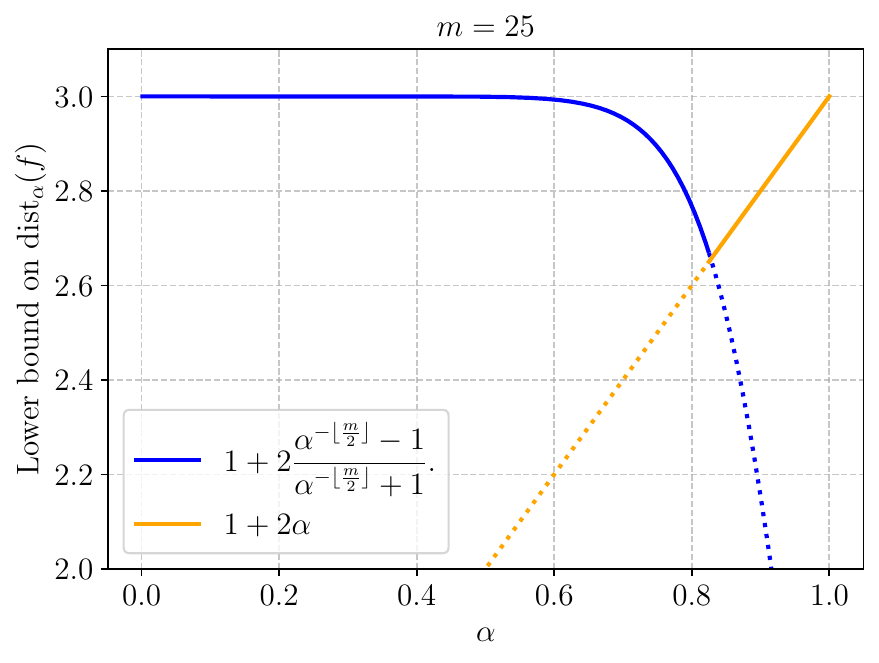}
   \label{fig:apx:25}
\end{minipage}
\hfill
\begin{minipage}{0.48\textwidth}
  \centering
  \includegraphics[width=\textwidth]{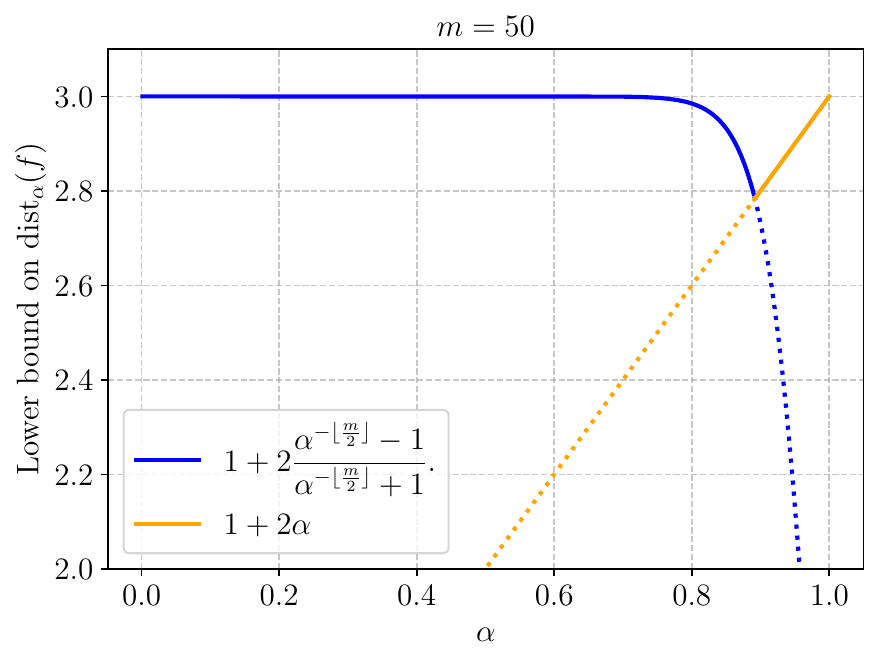}
   
\end{minipage}
   \caption{An illustration of the lower bounds established in \Cref{lem:mand_lb_left,lem:mand_lb_right} for 5, 10, 25 and 50 alternatives.}
   \label{fig:mand_lb2}
\end{figure}


\begin{restatable}{lemma}{LbA}
\label{lem:mand_lb_left}
    The metric distortion of any voting rule $f$ using ranked preferences with intensities ballot format is at least
    $1 + 2 \frac{1 - \alpha^{\lfloor\frac{m}{2}\rfloor}}{1 + \alpha^{\lfloor\frac{m}{2}\rfloor}}\cdot$
\end{restatable}

\begin{proof}
First, note that when $\alpha = 1$, the lower bound $1$ is trivial by definition of distortion. Also, we will consider the case $\alpha = 0$ separately. Hence, we may assume that $0 < \alpha < 1$.\\
Consider an instance with two agents and $m$ alternatives, where the first agent submits her vote as $a_1 \succ a_2 \succ \cdots \succ a_m$, and the second agent submits her vote in the reverse order, i.e., $a_m \succ a_{m-1} \succ \cdots \succ a_1$. Then, assume that the rule $f$ selects alternative $a$, and without loss of generality, $a$ appears in the first half of the ranking of the first agent. Now, consider the metric space $d$ in which alternatives are at unit distance from agent $1$, but for $b \in \alts$, $$d(2, b) = \alpha^{-\min(\rank{2}(b) - 1, \lfloor \frac{m}{2}\rfloor)} \cdot \frac{2}{\alpha^{-\lfloor \frac{m}{2}\rfloor} - 1}.$$ We claim that this distance function satisfies the triangle inequality and is consistent with  \(\prefp\).
Regarding consistency, for the second agent, $\frac{1}{\alpha}d(2, \pi_2(i)) = d(2, \pi_2(i+1))$ for $1 \leq i \leq \lfloor\frac{m}{2} \rfloor$. Hence, it is straightforward to verify that $d$ is consistent with $\pref_1$. 
Therefore, we only need to ensure that $d$ satisfies the triangle inequality, which means that for every \(b \neq c\) and \(i \neq j\),
\begin{equation}
\label{tri ineq}
d(i, b) \leq d(i, c) + d(j, c) + d(j, b).
\end{equation}
Note that $d(i, c) + d(j, c) + d(j, b)
 \geq 1+1+\frac{2}{\alpha^{-\lfloor \frac{m}{2}\rfloor } - 1}  =\alpha^{ - \lfloor \frac{m}{2}\rfloor} \cdot \frac{2}{\alpha^{-\lfloor \frac{m}{2}\rfloor} - 1} $, which is the maximum possible value of $d(i,b)$. This implies that $d$ satisfies the triangle inequality. Therefore \(\prefp \scns d\).

Now, we use this metric to give a bound on the distortion, we have:
\begin{align*}
\dist_\alpha(a, \elc) & \geq \frac{\scs(a)}{\scs(a_{m})} = \frac{d(1, a) + d(2, a)}{d(1, a_{m}) + d(2, a_{m})} \geq \;\;  \frac{1 + d(2, a_{\lfloor \frac{m+1}{2}\rfloor})}{1 + d(2, a_{m})} & \left(\text{Since } \rank{1}(a) \leq \lfloor \frac{m+1}{2}\rfloor \right)\\
&=  \frac{1 + \frac{2 \alpha^{-\lfloor \frac{m}{2}\rfloor}}{\alpha^{-\lfloor \frac{m}{2}\rfloor} - 1}}{1 + \frac{2}{\alpha^{-\lfloor \frac{m}{2}\rfloor} - 1}} = \frac{3 + \frac{2}{\alpha^{-\lfloor \frac{m}{2}\rfloor} - 1}}{1 + \frac{2}{\alpha^{-\lfloor \frac{m}{2}\rfloor} - 1}} = 1 + 2 \frac{1 - \alpha^{\lfloor\frac{m}{2}\rfloor}}{1 + \alpha^{\lfloor\frac{m}{2}\rfloor}}\cdot  
\end{align*}
Therefore, $\dist_\alpha(f)$ is at least $1 + 2 \frac{1 - \alpha^{\lfloor\frac{m}{2}\rfloor}}{1 + \alpha^{\lfloor\frac{m}{2}\rfloor}}\cdot$\\
Now, assume $\alpha = 0$. Consider an instance with two agents $\{1, 2\}$ and two alternatives $\{a, b\}$. First agent submits $a \succ b$ and the second one submits $b \succ a$. Without loss of generality assume that $f$ selects $a$ as the winner. Now, consider the metric $d$ such that $d(1, a) = d(1, b) = 1$, and $d(2, b) = 0, d(2, a) = 2$. Clearly, $d$ satisfies triangle inequality as well as intensity constraint. Also, $\frac{\scs(a)}{\scs(b)} = 3$. Therefore, $f$ has distortion at least $3$, as desired. 
\end{proof}

\begin{restatable}{lemma}{mandLbRight}
    \label{lem:mand_lb_right}
    The metric distortion of any voting rule $f$ using ranked preferences with intensities ballot format is at least
    $1 + 2\alpha$.
\end{restatable}

\begin{proof}
Fix $k$ and $m$ such that $k \leq \lfloor \frac{m}{2}\rfloor$ and consider an election with two agents and $m$ alternatives. 
Agents report intensity only in the $k\textsuperscript{\text{th}} $ position of their preferences. The first agent submits her vote as $a_1 \succ a_2 \succ \ldots \succ a_k \;\ggc\; a_{k+1} \succ \ldots \succ a_m$, and the second agent submits her vote in the reverse order, i.e., $a_m \succ a_{m-1} \succ \ldots \succ a_{m-k+1} \;\ggc\; a_{m-k} \succ \ldots \succ a_1$.

Let $f$ be a deterministic social choice rule and define $a := f(\prefp)$. Regarding the preferences, for at least one agent, $a$ appears after $\qggc$. Without loss of generality, we can assume that this happens for the first agent. Consider the following metric:
\begin{equation*}
d(1 ,b)  =
\begin{cases}
        0  &  \rank{1}(b) \leq k\\
        1 + \alpha^{-1} & \rank{1}(b) > k\\
\end{cases}, \text{ and }
d(2 ,b)  =
\begin{cases}
        1  & \rank{2}(b) \leq k\\
        \alpha^{-1} & \rank{2}(b) > k\\
\end{cases}.
\end{equation*}

Since $k \leq \lfloor \frac{m}{2}\rfloor$ and metric has only four distinct values, it is straightforward to verify that $d$ satisfies the triangle inequality and furthermore, $\prefp \scns d$.
\\
Note that $\scs(a_1) = d(1, a_1) + d(2, a_1) = \alpha^{-1}$, and $\scs(a) = d(1, a) + d(2, a) = 1 + \alpha^{-1} + d(2, a_l) \geq 2 + \alpha^{-1}$. Hence, we have:
\begin{align*}
    \dist_\alpha(a, \elc) & \geq 
    \frac{\scs(a)}{\scs(a_1)} \geq \frac{2 + \alpha^{-1}}{\alpha^{-1}} = 1 + 2\alpha.
\end{align*}
Therefore, $\dist_\alpha(f)$ is at least $1 + 2\alpha$. \qedhere
\end{proof}

\subsection{Upper Bound}
To complement our lower bounds, in this section, we present a voting rule that achieves a distortion better than 3. To do so, we start by analyzing the special case of moderate-up-to-$k$ preferences. 


\begin{definition}
    \label{def:k-indec}
We say a preference \( \pref = (\pi, \intns) \) is moderate-up-to-\(k\), if it exhibits the first intensity at the \( k+1\textsuperscript{\text{th}} \) position, i.e., \( \intns(k+1) = \qggc \) and \( \forall i \in [k]: \intns(i) = \qsucc \). If no \( \qggc \) appears in \( \intns \), the preference is considered as moderate-up-to-\((m-1)\). A preference profile is called moderate-up-to-\(k\) if all agents have moderate-up-to-\(k\) preferences.

\end{definition}

Our first goal is to prove an upper bound on the optimal metric distortion over moderate-up-to-$k$ preference profiles. Define 
\begin{align}
  w_{k} &=
  \begin{cases}
    \frac{\alpha+1}{3\alpha+1} & k = 1 \\
    1-\frac{2\alpha}{(1-\alpha)  t_{k-1} + 2  w_{k-1} + 2  \alpha} & k > 1
  \end{cases}, \text{ and} \label{first rec} \\
    t_{k} &=
    \begin{cases}
        \frac{1 - \alpha}{3\alpha + 1} & k=1\\
        w_{k} + (1 - w_{k})  t_{k-1} & k>1
    \end{cases}\label{second rec}.
 \end{align}
 \noindent Now we can state the main theorem as follows:

\begin{theorem}
\label{thm:mand_ub}
    There exists a rule $f$ that for $k \in [m - 1]$, given a moderate-up-to-$k$ preference profile, finds an alternative with metric distortion of at most $2+\max(\alpha, t_k).$
\end{theorem}

To prove this theorem, we adopt the method used by \citet{GHS20} to establish the existence of a deterministic rule with a metric distortion of 3. Their approach involves the $(\vec p, \vec q)$-domination graph, which is a bipartite graph defined for each alternative as follows.

\begin{definition}
\label{pq Graph}
Given an election, \(\elc = (\ags, \alts, \prefp)\), and vectors $\vec p\,\in\,\Delta(\ags)$ and $\vec q\,\in\,\Delta(\alts)$, the $(\vec p,\vec q)-$domination graph of an alternative $a$ is the vertex-weighted bipartite graph $$G_{\vec p, \vec q}^{\elc}(a) = (\ags, \alts, E_{a}, \vec p, \vec q),$$ where $(i, c) \in E_{a}$ if and only if $a \succ_i c$. The vertex corresponding to $i \in \ags$ has weight $p(i)$ and vertex $c \in \alts$ has weight $q(c)$. 
\end{definition}

\begin{definition}
\label{frac match}
We say $G_{\vec p, \vec q}^{\elc}(a)$ admits a fractional perfect matching, if there exists a weight function $w: E_{a} \to \mathbf{R}_{\ge0}$ such that the total weight of the edges incident to each vertex equals the weight of the vertex, i.e., for each $i \in \ags$ we have $\sum_{c: (i, c)\in E_{a}} w(i,c) = p(i)$, and for each $c \in \alts$ we have $\sum _{i: (i,c) \in E_{a}} w(i,c) = q(c)$.    
\end{definition}

By proving the following theorem, known as the Ranking-Matching Lemma, \citet{GHS20} show that such a matching exists. This forms the core of the proof of existence of a voting rule with a metric distortion 3.

\begin{theorem}{\label{thm:rml} (Ranking-Matching Lemma)}
For any election $\elc = (\ags, \alts, \prefp)$, and vectors $\vec p \in \Delta (\ags)$ and 
$\vec q \in \Delta (\alts)$, there exists an alternative a whose $(\vec p,\vec q)-$domination graph $G_{\vec p, \vec q}^{\elc}(a)$ admits a fractional perfect matching.
\end{theorem}

\citet{GHS20} apply \Cref{thm:rml} with $\vec{p} = (1/n, \ldots, 1/n)$ and $q(a) = \mathsf{Plu}(a)/n$ to show that the Plurality Matching rule achieves a distortion of 3. However, to prove \Cref{thm:mand_ub}, we need a generalized form of their result. To achieve this, we introduce a new class of voting rules called \textbf{Positional Scoring Matching rules}. These rules utilize a scoring vector of length $m$ and calculate the total score of each alternative based on the rankings. We define $\vec{q}$ based on the total score of the alternatives and set $\vec{p}$ to be uniform, and then, identify an alternative whose $(\vec{p}, \vec{q})$-domination graph admits a fractional perfect matching. The definition and analysis of this new class of rules may also be of independent interest for future studies.

\begin{definition}
    Given a unit-sum scoring vector $\vec{s} = (s_1, s_2, \ldots, s_m)$, the Positional Scoring Matching rule $\psm_{\vec{s}}$ is defined as follows:
    For an election $\elc = (\ags, \alts, \prefp)$, set $\vec{p} = (1/n, \ldots, 1/n)$. For each  $a \in \alts$, calculate $q(a) = \sum_{i \in \ags} s_{\rank{i}(a)}/n$. The rule then selects the alternative $a^*$ for which $G_{\vec{p}, \vec{q}}^{\elc}(a^*)$ admits a fractional perfect matching.
\end{definition}

Note that, according to \Cref{thm:rml}, such an alternative always exists. It is also important to mention that the Plurality Matching rule, which was the first deterministic rule to achieve a distortion of 3, is a member of this class with $\vec{s} = (1, 0, \ldots, 0)$. We now analyze the distortion of rules within this class when applied to moderate-up-to-$k$ preference profiles.

\begin{restatable}{lemma}{scoringMatching}
\label{lem:scoring-matching}
    For any unit-sum weight vector $\vec r = (r_1, r_2, \ldots, r_{k+1}, 0, \ldots, 0)$ of length $m$, and moderate-up-to-$k$ preference profile $\prefp$,  $\dist_\alpha(\psm_{\vec r}(\prefp), \prefp)$ is upper bounded by 
    $$ 2+\max\!\left(\!\alpha,\!\max_{j \in [k+1]}\left(\displaystyle\sum\limits_{i \in [k+1] \land i \neq j }
        \!r_{i}\,\max\left(1, \alpha ^{j-i}\right) \right)\!-\!r_{j}\right).$$
\end{restatable}

\begin{proof}
     Let $\vec p$ be the uniform probability distribution that assigns probability $1/n$ to each agent in $N$. We define $a := \psm_{\vec r}(\prefp)$ and $b := \opt_{d}$ for an arbitrary metric $d$. Recall that $\vec{p}$ and $\vec{q}$ are related to $\psm_{\vec{r}}$, and $a$ is the alternative whose $(\vec{p}, \vec{q})$-domination graph admits a fractional perfect matching, with $w$ as its weight function. For each vertex $x \in \ags \cup \alts $ in $G_{\vec p, \vec q}^{\elc}(a)$, we denote by $R(x)$ the set of vertices adjacent to $x$ in the undirected version of $G_{\vec p, \vec q}^{\elc}(a)$. We have:
    \begin{align} 
    \sum\limits_{i \in \ags} d(i,a)\, & \leq\,\sum\limits_{i \in \ags}\,
    \sum\limits_{c \in  R(i)} 
    \frac{w(i,c)}{\sum\limits_{c' \in R(i)} w(i,c')}\,d(i,c)\ & \left(d(i, a) \leq d(i,c) \text{ for all $c \in R(i)$}\right)
    \nonumber  
    \\ 
    & =\,\sum\limits_{c\in\alts}
    \sum\limits_{i \in R(c)} 
    \frac{w(i,c)}{\sum\limits_{c' \in R(i)} w(i,c')}\,d(i,c)\ & \left(\text{rearranging the sum}\right)
    \nonumber
    \\ 
    & =\,\sum\limits_{c \in \alts }\,
    \sum\limits_{i \in R(c)}
    \frac{w(i,c)}{1/n}\,d(i,c) 
    \nonumber
    & \left(\sum\limits_{c \in R(i)} w(i,c) = p(i) = \frac{1}{n}\right)
    \\ 
    & \leq\,\sum\limits_{c \in \alts}\,
    \sum\limits_{i \in R(c)}
    \frac{w(i,c)}{1/n}\,\left(d(i,b)+d(b,c)\right) 
    \nonumber
    & \left(\text{triangle inequality}\right)
    \\ 
    \nonumber
    & = \sum\limits_{c \in \alts}\,
    \sum\limits_{i \in R(c)}
    \frac{w(i,c)}{1/n}\, d(i,b) + 
    \sum\limits_{c \in \alts}\,
    \sum\limits_{i \in R(c)}  n\,w(i,c)\,d(b,c)
    \\ 
    \nonumber
    & =\sum\limits_{i \in \ags} \,
    \sum\limits_{c \in R(i)}
    \frac{w(i,c)}{1/n}\, d(i,b) + 
    \sum\limits_{c \in \alts}\,
    \sum\limits_{i \in R(c)}  n\,w(i,c)\,d(b,c)
    & \left(\text{rearranging the first sum}\right)
    \\ 
    \label{upp partition}
    & =\sum\limits_{i \in \ags} d(i,b) + 
    \sum\limits_{c \in \alts}\,
    \sum\limits_{i \in R(c)}  n\,w(i,c)\,d(b,c)
    & \left(\sum_{c \in R(i)} w(i,c) = p(i) = \frac{1}{n}\right)
    \\ 
    \nonumber
    & =\sco(b) + 
    \sum\limits_{c \in \alts} n\, q(c)\,d(b,c)
    & \left(\sum_{i \in R(c)} w(i,c) = q(c)\right)
    \\ 
    \nonumber
    & =\sco(b) + \sum\limits_{c \in \alts}\,
    \sum\limits_{i \in \ags}r_{\rank{i}(c)}\,d(b,c)
    & \left(q(c) = \sum_{i \in \ags} \frac{r_{\rank{i}(c)}}{n} \right) 
    \\ 
    & =\sco(b) + \sum\limits_{i \in \ags}\,
    \sum\limits_{j=1}^{k+1}r_{j}\,d(b,\pi_i(j)).
    \nonumber
    & \left(\text{rearranging the sum}\right)
    \end{align}
    Now, we partition the summation over all agents into separate sums based on the position of $b$ in their preferences:
    \begin{align*}
    \sum\limits_{i \in \ags} d(i,a)\,& \leq \sco(b) +\sum\limits_{l=1}^{k+1}\, \sum\limits_{\substack{i:\\\pi_i(l) = b}}\,
\sum\limits_{j=1}^{k+1}r_{j}\,d(b,\pi_i(j))
    + \sum\limits_{\substack{i:\\ \rank{i}(b)\geq k+2}}\,
\sum\limits_{j=1}^{k+1}r_{j}\,d(b,\pi_i(j))
    \\ 
    & = 
    \sco(b) + \sum\limits_{l=1}^{k+1}\,
    \sum\limits_{\substack{i:\\ \pi_i(l) = b}}\,
    \sum\limits_{j=1}^{l-1}
    r_{j}\, d(b,\pi_i(j)) +
    \sum\limits_{l=1}^{k+1}
    \sum\limits_{\substack{i:\\ \pi_i(l) = b}}\!
    \sum\limits_{j=l+1}^{k+1}
    r_{j}  d(b,\pi_i(j))
    \\ &
    +\!\sum\limits_{\substack{i:\\ \rank{i}(b)\geq k+2}} \!
    \sum\limits_{j=1}^{k+1}
    r_{j}  d(b,\pi_i(j))
    & \left( \text{$d(b, b) = 0$} \right)
    \\ 
    & \leq
    \sco(b) +
    \sum\limits_{l=1}^{k+1}\,
    \sum\limits_{\substack{i:\\ \pi_i(l) =b}}\,
    \sum\limits_{j=1}^{l-1}
    r_{j}\left( d(i, b)+ d(i, \pi_i(j) )\right) 
    \\ & \qquad \quad \;+
    \sum\limits_{l=1}^{k+1}\,
    \sum\limits_{\substack{i:\\ \pi_i(l) =b}}\,
    \sum\limits_{j=l+1}^{k+1}
    r_{j}\left( d(i, b)+ d(i,\pi_i(j))\right) 
    \\ & \qquad \quad \;+
    \sum\limits_{\substack{i:\\ \rank{i}(b)\geq k+2}}\,
    \sum\limits_{j=1}^{k+1}
    r_{j} \, \left( d(i, b)+ d(i,\pi_i(j))\right)
    & \left( \text{triangle inequality}\right)
    \\ 
    & \leq
    \sco(b) +
    \sum\limits_{l=1}^{k+1}\,
    \sum\limits_{\substack{i:\\ \pi_i(l) = b}}\,
    \sum\limits_{j=1}^{l-1}
    2r_{j}\,d(i, b) & \left(\text{$\pi_i(j)\succ_i b$}\right)
    \\ & \qquad \quad \;+
    \sum\limits_{l=1}^{k+1}\,
    \sum\limits_{\substack{i:\\ \pi_i(l) =b}}\,
    \sum\limits_{j=l+1}^{k+1}
    r_{j}\left(1+\frac{1}{\alpha^{j-l}}\right)d(i, b)
    \\ & \qquad \quad \;
    +\sum\limits_{\substack{i:\\ \rank{i}(b)\geq k+2}}\,
    \sum\limits_{j=1}^{k+1}
    r_j \, (1+\alpha)\, d(i, b)
    & \left(\text{$\qsucc$ and $\qggc$ constraint}\right) 
    \\ & = 
    \sco(b) +
    \sum\limits_{l=1}^{k+1}\,
    \sum\limits_{\substack{i:\\ \pi_i(l)=b }}\,
    \left(
    \sum\limits_{j=1}^{l-1}
    2r_{j}\,d(i, b) 
    +
    \sum\limits_{j=l+1}^{k+1}
    r_{j}\left(1+\frac{1}{\alpha^{j-l}}\right)d(i, b)
    \right)
    \\ & \qquad \quad \;+
    (1+\alpha) \sum\limits_{\substack{i:\\ \rank{i}(b)\geq k+2}}d(b,i)
    & \left(\text{$\sum_{i = 1}^{k + 1} r_i = 1$}\right)
    \\ 
    & =\sco(b) +(1+\alpha)\,\sum\limits_{\substack{i:\\ \rank{i}(b) \geq k+2}}d(i, b) 
    \\ & \qquad \quad \; + \sum\limits_{l=1}^{k+1}\, \sum\limits_{\substack{i:\\\pi_i(l) = b}}\,
    d(b,i) \sum\limits_{\substack{j: j \neq l \\ 1\leq j \leq k+1}}
    r_{j}\left(1+\alpha^{\min\left(0,l-j\right)}\right) \\
    & =\sco(b) +(1+\alpha)\,\sum\limits_{\substack{i:\\ \rank{i}(b) \geq k+2}}d(i, b) 
    \\ & \qquad \quad \;
    +\sum\limits_{l=1}^{k+1}\, \sum\limits_{\substack{i:\\\pi_i(l) = b}}\,
    d(i, b)
    \left( 1- r_{l} +  \sum\limits_{\substack{j: j \neq l \\ 1\leq j \leq k+1}}
    r_{j}\,\alpha^{\min\left(0,l-j\right)}
    \right).
    \end{align*}
    If we choose $\beta$ such that
    \begin{equation}
     \label{beta}
     \max_{l \in [k+1]} \displaystyle\sum\limits_{\substack{j: j \neq l \\ 1\leq j \leq k+1}}
        r_{j}\,\alpha^{\min\left(0,l-j\right)} - r_{l} = \max_{l \in [k+1]} \displaystyle\sum\limits_{\substack{j: j \neq l \\ 1\leq j \leq k+1}}
        r_{j}\,{\max\left(1,\alpha^{l-j}\right)} - r_{l}=\beta,   
    \end{equation}
    then we have:
    \begin{align}
        \sco(a) \leq 
        \sco(b) +\sum\limits_{l=1}^{k+1}\, \sum\limits_{\substack{i:\\ \pi_i(l) = b}}\,
        d(i, b)
        \left(1+ \beta \right)
        +(1+\alpha)\,\sum\limits_{\substack{i:\\\rank{i}(b)\geq k+2}}d(i, b) 
        \nonumber 
        \\ 
        \leq  
        \sco(b) + 
        \left( 1+\,\max \left (\beta,\,\alpha \right)\right)
        \displaystyle\sum\limits_{i \in \ags} d(i, b) 
        \leq 
        \left( 2+\,\max \left (\beta,\,\alpha \right)\right)
        \sco(b).
        \nonumber
    \end{align}
    Hence, the distortion of $a$ is at most $2 + \max \left(\beta, \alpha \right)$.
\end{proof}

This analysis gives us the power to optimize the scoring vector $\vec r$ to get lower distortion. We can formulate our optimization problem as following:
$$\minimize_{\vec r  \in \Delta( [k+1])} \max_{j \in [k+1]}\left(\displaystyle\sum\limits_{i \in [k+1] \land i \neq j }
        r_{i}\,\max\left(1, \alpha ^{j-i}\right)\right) - r_{j},$$ which in the matrix form can be written as 
        \begin{equation}
        \minimize_{\vec r \in \Delta([k + 1])} \max_{j \in [k + 1]} \left( \vec {r}^{\intercal} M\right)_{j},\label{eq:opt-min-max}\end{equation}
        
      \noindent  where 
      \vspace{-2mm}
      $$M=
    \begin{bmatrix}
    -1 & 1 & \dots &  1 \\
    \frac{1}{\alpha} & -1 & \dots & 1\\
     \vdots &\vdots  &\ddots  & \vdots & \\
    \frac{1}{\alpha^{k}} & \frac{1}{\alpha^{k-1}} &  \dots \frac{1}{\alpha}  & -1 
\end{bmatrix}.$$

Optimizing Objective \eqref{eq:opt-min-max} is analogous to finding the value of a two-player zero-sum game with payoff matrix $M$ for the second player. 
    Assuming $\vec{x}$ and $\vec{y}$ represent the mixed strategies of the first (row) player and the second (column) player, respectively, then $\vec{x}^{\intercal} M \vec{y}$ represents the payoff of the second player. In Lemma \ref{lem:mixed}, we introduce an equilibrium of this game and the corresponding payoff.

\begin{restatable}{lemma}{lemMixed}
\label{lem:mixed}
For each $k\!\in\![m\!-\!1]$ consider the mixed strategy $\vec r^k = (w_k, w_{k-1}(1-w_k), \ldots,\!w_1\Pi_{j =2}^k (1-w_j),\!\Pi_{j =1}^k\!(1-w_j))$. Then $(\!\vec\!r^k,\!(\!\vec r^k\!)^R)$ is an equilibrium of the zero-sum game with payoff matrix $M$, and the payoff of this equilibrium, which is equal to the optimal value of Objective \ref{eq:opt-min-max}, is $t_{k}$ ($r^{R}$ is vector $r$ with reverse order of indices).
\end{restatable}

\begin{proof}
We denote by $M^k$ the payoff matrix $M$ corresponding to the indecisive-up-to-$k$ preferences. First, we prove by induction that $({\vec {r}^k})^\intercal  M^k = t_k \mathbf{1}^\intercal$, where $\mathbf{1}$ is the vector with all entries equal to one. One can easily check that the base case holds trivially. Hence, we assume that $({\vec {r}^{k - 1}})^\intercal  M^{k-1} = t_{k-1} \mathbf{1}^\intercal$, and aim to show the equality for $k$. Note that matrix $M^k$ has the following form:  
\begin{equation}
 \label{M^k: 1}
     M^{k} = 
     \begin{bmatrix}
         -1 & \mathbf{1}^\intercal \\
         \vec{\alpha}^{k} & M^{k-1}
     \end{bmatrix},
 \end{equation}
 where $\vec{\alpha}^{k}$ is the vector with $i\textsuperscript{\text{th}}$ element equal to $\frac{1}{\alpha ^i}$. Moreover, based on the definition of $\vec{r}^k$, we have:
\begin{equation}
  \label{r^k: 1}
  (\vec{r}^{k})^\intercal = [w_{k} \;\; (1-w_{k})(\vec{r}^{k-1})^{\intercal}].
\end{equation}
Now, for every $2 \leq i \leq k + 1$, we can write ($M^{k}_{i}$ is the $i\textsuperscript{\text{th}}$ column of $M^{k}$): 
\begin{align}
\nonumber
(\vec{r}^{k})^\intercal M^{k}_{i} &= w_{k} + (1 - w_{k}) (\vec{r}^{k-1})^\intercal M^{k-1}_{i-1} & \text{(Equations \eqref{M^k: 1} and \eqref{r^k: 1})} \\
\nonumber
&= w_{k} + (1 - w_{k})t_{k-1} & \text{(Induction hypothesis)}\\
\nonumber
&= t_{k}. & \text{(Equation \eqref{second rec})}
\end{align}
Therefore, the elements from the second to the $(k+1)\textsuperscript{\text{th}}$ of $(\vec{r}^{k})^\intercal M^{k}$ are equal to $t_k$. Hence, we only need to check its first element, i.e., $(\vec{r}^{k})^\intercal M^{k}_{1}$. Note that we can write the first column of $M^k$ as follows ($\vec{e_i}^{\intercal}$ is a vector whose only non-zero element is the $i\textsuperscript{\text{th}}$ element, which is equal to one): 
 \begin{equation}
 \label{M^k: 2}
     \begin{split}
         (M^{k}_{1})^\intercal &= [ -1 \;\;\; \frac{1}{\alpha} (M^{k-1}_{1})^\intercal + \frac{2}{\alpha} \vec{e_1}^{\intercal}].
     \end{split}
 \end{equation}
 We have:
 \begin{align}
    \nonumber
     (\vec{r}^{k})^\intercal  M^{k}_{1} & = -w_{k} + (1 - w_{k})   \frac{1}{\alpha} (\vec{r}^{k-1})^\intercal  M^{k-1}_{1} + \frac{2}{\alpha} (1 - w_{k}) w_{k-1} & (\text{Equations ~\eqref{r^k: 1} and \eqref{M^k: 2}})
     \nonumber
     \\& = -w_{k} + (1- w_{k})(\frac{1}{\alpha}t_{k-1} + \frac{2}{\alpha}w_{k-1}) & \nonumber (\text{Induction hypothesis}) \\
     \nonumber
     &= w_k + (1 - w_k)t_{k-1} - 2 + (1 - w_k)((\frac{1}{\alpha} - 1)t_{k-1} + \frac{2}{\alpha} w_{k-1} + 2) \\
     &= w_k + (1 - w_k)t_{k-1}  &\nonumber (\text{Equation \eqref{first rec}})\\
     &= t_k. & \nonumber (\text{Equation \eqref{second rec}})
 \end{align}
 Now, we prove that $(\vec{r}^k, (\vec{r}^k)^R)$ is an equilibrium of the zero-sum game with payoff matrix $M^k$. Observe a key property of $M^k$: for every $1 \leq i \leq j \leq k+1$, $M^k(i, j) = M^k(k + 2 - j, k + 2 - i)$. This means that $M^k$ is symmetric with respect to its minor diagonal. For such a matrix, we can deduce that $(\vec {r}^k)^\intercal  M^k = (M^k {(\vec {r}^k)^R})^R.$ But we already know that $(\vec {r}^k)^\intercal  M^k = t_k \mathbf{1} ^ \intercal$. Hence, $M^k {(\vec {r}^k)^R} = t_k \mathbf{1}$. Now, if the first player plays ${\vec {r}^k}$ and the second one plays $({\vec {r}^k})^R$, neither player will be willing to  change their strategy, since for every $\vec{x}, \vec{y} \in \Delta([k+1])$, we have:
 \[(\vec {r}^k)^\intercal M^k \vec{y} = t_k \mathbf{1}^\intercal \vec{y} = t_k,
 \]
 \[
 \vec{x} M^k ({\vec {r}^k})^R = \vec{x} t_k \mathbf{1} = t_k.
 \]
 As a result, $(\vec{r}^k, (\vec{r}^k)^R)$ forms an equilibrium of the game, and $t_k$ is the value of the game, as desired. 
\end{proof}

Now we are ready to prove \Cref{thm:mand_ub}.
\begin{proof}[Proof of \Cref{thm:mand_ub}]
Let $\vec s = (r^k_1, \ldots, r^k_{k+1}, 0, \ldots,0)$, where $\vec r^k = (r^k_1, \ldots, r^k_{k+1})$ is the optimal vector we found in Lemma \ref{lem:mixed}, and consider the rule $\psm_{\vec s}$.  Applying Lemma \ref {lem:scoring-matching} gives us an upper bound on $\dist_{\alpha}(\psm_{\vec s})$, and Lemma \ref{lem:mixed} simplifies this bound to $2 + \max(\alpha, t_{k})$.
\end{proof}

Thus far, our results apply to instances with moderate-up-to-\(k\) preference profiles. However, in the general case, agents may have their first \(\qggc\) at different ranks. Let \(\ell_i\) denote the intensity rank of agent \(i\), meaning agent \(i\) is moderate-up-to-\(\ell_i\). The following theorem generalizes the result of \Cref{thm:mand_ub} to accommodate agents with varying intensity ranks. The bound provided depends on \(\ell_{\max}\), the maximum intensity rank among all agents, i.e., \(\ell_{\max} := \max_{i \in \ags} \ell_i\). Note that the bound provided in \Cref{gen upp} is independent of $m$. You can find an illustration of the bound provided in \Cref{gen upp} for different values of $\alpha$ and $\ell_{\max}$ in \Cref{fig:mand_k_1,fig:mand_k_2}.

\begin{restatable}{theorem}{genUp}
\label{gen upp}
 There exists a rule $f$ such that, given any preference profile, finds an alternative with metric distortion of at most $2+\max(\alpha, t_{\ell_{\max}})$.
\end{restatable}

\begin{proof}
Denote by $\ags_h$ the set of agents with intensity rank $h$. Furthermore, let $H \subset \mathbb{N}$ be a finite set such that $\bigcup_{h \in H} N_h = N$. Also,
let $\vec{r^{\ell_i}}$ be the optimal vector for $\ell_i$-indecisive elections, which we present in Lemma \ref{lem:mixed}. Let $\vec{p}$ be the uniform probability distribution that assigns probability $1/n$ to each agent in $N$. Furthermore, for alternative  $c \in A$, we define:\[ q(c) = \sum_{i \in \alts} \frac{{r^{l_i}}_{\rank{i}(c) } }{n} \cdot \] Using Theorem \ref{thm:rml}, there is an alternative $a$ whose $\left(\vec{p},\vec{q}\right)$-domination graph admits a fractional perfect matching. Now, we want to prove an upper bound for $\sum_{i \in N} d(i, a)$. We define $b := \opt_d$  for an arbitrary metric $d$. For each vertex $x \in \ags \cup \alts $ in $G_{\vec p, \vec q}^{\elc}(a)$, we denote by $R(x)$ the set of vertices adjacent to $x$ in the undirected version of $G_{\vec p, \vec q}^{\elc}(a)$. We can use a similar argument to the one in the proof of Lemma \ref{lem:scoring-matching} up to Equation \eqref{upp partition}. We have:
\begin{align}
    \nonumber
    \sum\limits_{i \in \ags} d(i,a)&\leq \sum\limits_{i\,\in\, \ags} d(i,b) + 
    \sum\limits_{c\,\in\,\candid}\,
    \sum\limits_{i\,\in\,R(c)}  n\,w(i,c)\,d(b,c)
    \\ 
    \nonumber
    & =\sco(b) + 
    \sum\limits_{c\,\in\,\candid} n\, q(c)\,d(b,c)
    & \left(\sum_{i \in R(c)} w(i,c) = q(c)\right)
    \\ 
    \nonumber
    & =\sco(b) + \sum\limits_{c \in \candid}\,
    \sum_{i \in \ags}{r^{l_i}}_{\rank{i}(c) } d(b,c)
    & \left(q(c) = \sum_{i \in \ags} \frac{{r^{\ell_i}}_{\rank{i}(c) } }{n} \right) 
     \\ 
    & =\sco(b) + \sum\limits_{i \in \ags}\, \sum\limits_{j=1}^{l_i+1}{r^{\ell_i}}_{j}\,d(b,\pi_i(j))
    \nonumber
    & \left(\text{rearranging the sum}\right)
    \nonumber\\
    & \leq \sco(b) +\sum\limits_{h \in H}\sum\limits_{k=1}^{h+1}\, \sum\limits_{\substack{i: \pi_i(k)= b\\ i \in N_h}}\,
    \sum\limits_{j=1}^{h+1}r^h_j\,d(b,\pi_i(j))
    \nonumber 
    \\ & \qquad \quad \;+
    \sum\limits_{h\in H} \sum\limits_{\substack{i:i \in N_h\\ \rank{i}(b)\geq h+2}}\,
    \sum\limits_{j=1}^{h+1}r^h_j\,d(b,\pi_i(j))
    \nonumber
\\
    & = 
    \sco(b) +\sum\limits_{h \in H} \sum\limits_{k=1}^{h+1}\,
    \sum\limits_{\substack{i: \pi_i(k)= b\\ i \in N_h}}\,
    \sum\limits_{j=1}^{k-1}
    r^h_j\, d(b,\pi_i(j))
    \nonumber 
    \\ & \qquad \quad \;+
    \sum\limits_{h \in H} \sum\limits_{k=1}^{h+1}\,
    \sum\limits_{\substack{i: \pi_i(k)= b\\ i \in N_h}}\,
    \sum\limits_{j=k+1}^{i+1}
    r^h_j \, d(b,\pi_i(j))
    \nonumber 
    \\& \qquad \quad \;+
    \sum\limits_{h \in H} \sum\limits_{\substack{i: i \in N_h\\ \rank{i}(b)\geq h+2}}\,
    \sum\limits_{j=1}^{h+1}
    r^h_j \, d(b,\pi_i(j))
    & \left( \text{$d(b, b) = 0$} \right)
    \nonumber
    \\ 
    & \leq
    \sco(b) +
    \sum\limits_{h \in H}\sum\limits_{k=1}^{h+1}\,
    \sum\limits_{\substack{i:\pi_i(k) = b\\ i \in N_h}}\,
    \sum\limits_{j=1}^{k-1}
    r^h_j\left( d(i, b)+ d(i, \pi_i(j) )\right) 
    \nonumber
    \\ & \qquad \quad \;+
    \sum\limits_{h \in H} \sum\limits_{k=1}^{h+1}\,
    \sum\limits_{\substack{i: \pi_i(k)= b\\ i \in N_h}}\,
    \sum\limits_{j=k+1}^{h+1}
    r^h_j\left( d(i, b)+ d(i,\pi_i(j))\right) 
    \nonumber
    \\ & \qquad \quad \;+
    \sum\limits_{h \in H} \sum\limits_{\substack{i: i \in N_h\\ \rank{i}(b)\geq h+2}}\,
    \sum\limits_{j=1}^{h+1}
    r^h_j \, \left( d(i, b)+ d(i,\pi_i(j))\right)
    & \left( \text{triangle inequality}\right)
    \nonumber
    \\ 
    & \leq
    \sco(b) +
    \sum\limits_{h \in H}\sum\limits_{k=1}^{h+1}\,
    \sum\limits_{\substack{i:\pi_i(k) = b\\ i \in N_h}}\,
    \sum\limits_{i=1}^{k-1}
    2r^h_j\,d(i, b) & \left(\text{$\pi_i(h)\succ_i b$}\right)
    \nonumber
    \\ & \qquad \quad \;+
    \sum\limits_{h \in H} \sum\limits_{k=1}^{h+1}\,
    \sum\limits_{\substack{i: \pi_i(k)= b\\ i \in N_h}}\,
    \sum\limits_{j=k+1}^{i+1}
    r^h_j\left(1+\frac{1}{\alpha^{h-j}}\right)d(i, b) 
    \\ & \qquad \quad \;+
    \sum\limits_{h \in H} \sum\limits_{\substack{i: i \in N_h\\ \rank{i}(b)\geq h+2}}\,
    \sum\limits_{j=1}^{h+1}
    {r^h}_j \, (1+\alpha)\, d(i, b)
    & \left(\text{$\qsucc$ and $\qggc$ constraints}\right)
    \nonumber
    \\ 
    & =
    \sco(b) +\sum\limits_{h \in H}\sum\limits_{k=1}^{h+1}\,
    \sum\limits_{\substack{i:\pi_i(k) = b\\ i \in N_h}}\,
    \left(\sum\limits_{j=1}^{k-1}
    2r^h_j \,d(i, b) \right.
    + \nonumber
    \\ & \left.\qquad \qquad \qquad  \qquad\qquad\qquad
    \sum\limits_{j=k+1}^{h+1} r^h_j\left(1+\frac{1}{\alpha^{h-j}}\right)d(i, b)
    \right)
    \nonumber
    \\ & \qquad \quad \; +
    (1+\alpha) 
    \sum\limits_{h \in H} \sum\limits_{\substack{i: i \in N_h\\ \rank{i}(b)\geq h+2}}\,
    d(i,b)
    & \left(\text{$\sum_{j = 1}^{h + 1} {r^h}_i = 1$}\right)
    \nonumber
    \\ 
    &
    =\sco(b) + (1+\alpha) \sum\limits_{h \in H} \sum\limits_{\substack{i: i \in N_h\\ \rank{i}(b)\geq h+2}}\,d(i,b) 
    \nonumber
    \\ & \qquad \quad \;
    +\sum\limits_{h \in H}\sum\limits_{k=1}^{h+1}\,
    \sum\limits_{\substack{i:\pi_i(k) = b\\ i \in N_h}}\,
    d(i,b) \sum\limits_{\substack{j: j \neq k \\ 1\leq j \leq h+1}}
    r^h_j\left(1+\frac{1}{\alpha^{\max\left(0,j-k\right)}}\right).
    \nonumber
    \end{align}

From the proof of  Lemma \ref{lem:mixed} we know that for each $k\in [h+1] $,
$\sum\limits_{\substack{j: j\neq k\\ 1\leq j\leq k+1} } {r^i}_j \left(1+ \frac{1}{\alpha^{max\{0, j-k\}} } \right) \leq t_h + 1$ .
Then we have:
\begin{align}
\sum_{i \in N} d(i,a) 
& \leq sc(b) + 
\sum\limits_{h \in H}\,
\sum\limits_{k=1}^{h+1}
\sum\limits_{\substack{i:i \in N_h\\ \pi_i(k) =b}}\,
d(i, b) \left(t_h +1 \right) + \left(1+ \alpha \right)
\sum_{h \in H}  \sum\limits_{\substack{i:\\
\rank{i}(b)\geq h+2}} d(i,b)
\nonumber
\\
& \leq 
sc(b) + \sum_{h \in H} \sum_{i \in N_h} d(i,b) \left(1+ \max\{\alpha , t_h\}  \right).
\nonumber
\end{align}
We know that $t_{h}$ is increasing in $h$. Hence, if we set $ \ell_{\max} =\max_{h \in H}(h)$, we have:
\[\sum_{i \in N_{h}} d(i,a) \leq (2 + \max(t_{\ell_{\max}}, \alpha)) \sum_{i \in N_{h}} d(i,b), \]
for each $h$, and finally, we have:
\[
\scs(a) = \sum_{i \in \bigcup_{h \in H} N_{h}} d(i,a) \leq (2 + \max(t_{\ell_{\max}}, \alpha))\; \scs(b).
\qedhere \]
\end{proof}

\begin{figure}[t]
  \centering 
  \begin{minipage}{0.45\textwidth}
  \includegraphics[width=\textwidth]{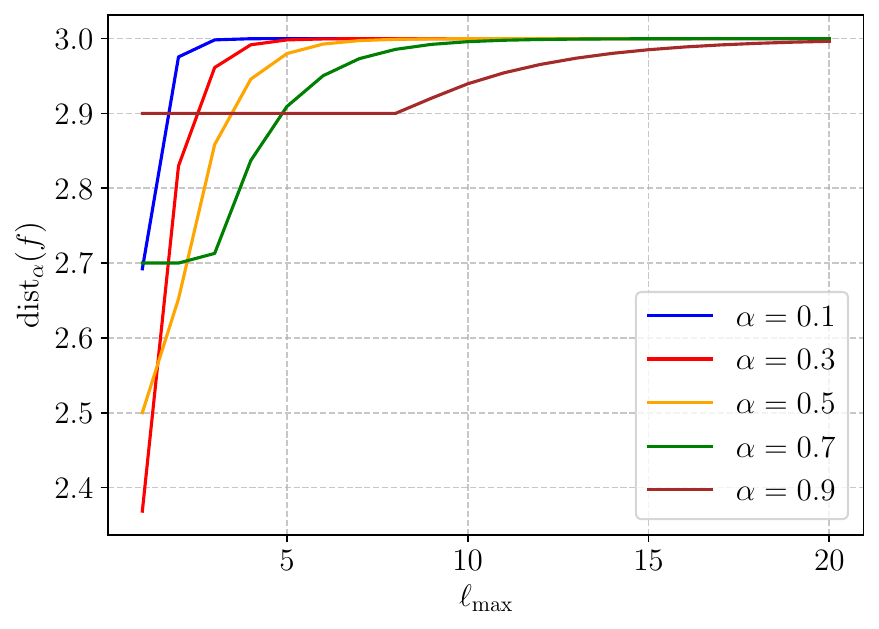}
  \caption{The distortion of the rule described in \Cref{gen upp} as a function of $\ell_{\max}$ for different  $\alpha$.}
  \label{fig:mand_k_1}
  \end{minipage}
  \hfill
  \begin{minipage}{0.45\textwidth}
  \centering
  \includegraphics[width=.95\textwidth]{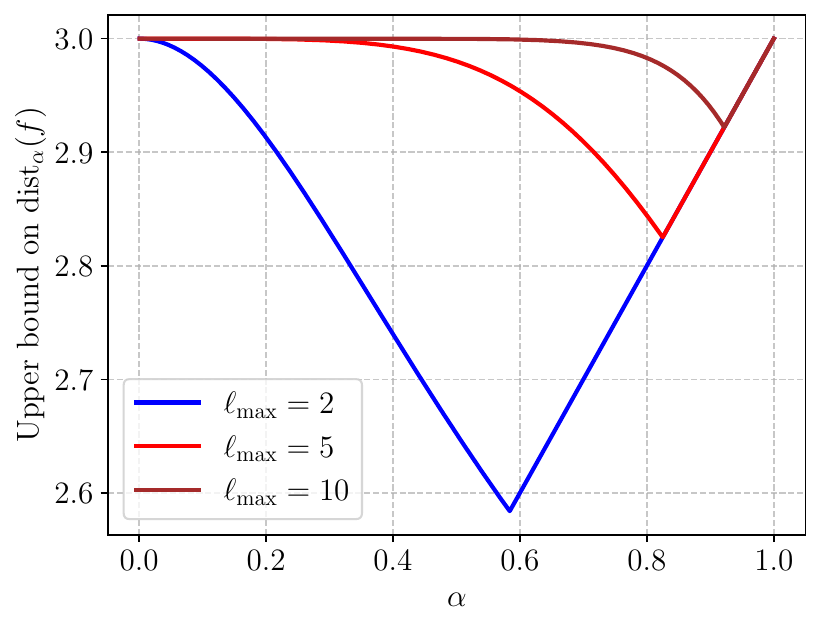}
   \caption{The distortion of the rule described in \Cref{gen upp} as a function of $\alpha$ for different $\ell_{\max}$.}
   \label{fig:mand_k_2}
\end{minipage}
\end{figure}

\Cref{gen upp} gives a strong upper bound on the distortion as a function of $\ell_{\max}$. 
The problem with this result is that it does not take into account the number of agents with high-intensity rank. For example, one agent has intensity rank of $m-1$ and all other agents report intensity in their first position.
In such cases, although the bound is still lower than $3$, it becomes dependent on $t_{m-1}$, which is undesirable. To address this, we demonstrate the robustness of our results, showing that the upper bound only slightly degrades if a constant portion of agents have a high-intensity rank.

\begin{restatable}{theorem}{robustness}
    \label{thm:robustness}
    Consider a voting rule $f$ that achieves distortion $D_\ell$ on preference profiles with maximum intensity rank of $\ell$ among the agents. We can design a rule $f_\ell$ such that for any $\beta \in [0, 1/2]$, the distortion of $f_\ell$ on any preference profile (with the same number of alternatives) in which the intensity rank of at most $\beta n$ agents is more than $\ell$ is at most
    $D_\ell + \frac{\beta}{1-\beta} (1+D_\ell).$
\end{restatable}

\begin{proof}
    Denote the agents with maximum intensity rank of $\ell$ as $N_1$ and the remaining agents as $N_2$. Suppose that $f$ outputs $a$ when is given the preference profiles of agents in $N_1$. Then we define the output of $f_\ell$ when is given all instances to be $a$. For an arbitrary metric $d$, we define $b := \opt_d$. Then, we have 
    \begin{equation}
        \label{dist on N_1}
        \sum\limits_{i \in N_1} d(i,a) \leq D_{\ell} \sum\limits_{i \in N_1} d(i,b).
    \end{equation} 
    \\ 
    We define $j := \text{arg min}_{i \in N_1}\left(d(i,b) + d(i,a)\right).$ Hence, we can write 
    \begin{align}
    \nonumber
         d(j,b) + d(j,a) &\leq \frac{1}{|N_1|} \sum\limits_{i \in N_1}\left(d(i,b) + d(i,a)\right)  & \text{(property of $j$)} \\ 
         \label{property of j}
         &\leq \frac{1}{\left(1-\beta\right)n} \sum\limits_{i \in N_1}\left(d(i,b) + d(i,a)\right). & \left(\left(1-\beta\right)n \leq |N_1|\right)
     \end{align}
    Now, we have 
    \begin{align}
        \nonumber
        \sum\limits_{i \in N_2} d(i,a) &\leq \sum\limits_{i \in N_2} \left(d(i, b) + d(j, b) + d(j, a)\right) & \text{(triangle inequality)}\\
        \nonumber
        & \leq \sum\limits_{i \in N_2} d(i,b) + \sum\limits_{i \in N_2} \frac{1}{\left(1-\beta\right)n} \sum\limits_{i' \in N_1}\left(d(i',b) + d(i',a)\right)  & \left(\text{Inequality \eqref{property of j}}\right)\\
        \label{bounds on dist_{N_2}}
        &\leq \sum\limits_{i \in N_2} d(i,b) + \frac{\beta}{1-\beta} \sum\limits_{i \in N_1}\left( d(i,b) + d(i,a) \right). & \left(|N_2| \leq \beta n\right)
    \end{align}
    Therefore, we can conclude that 
    \begin{align}
        \nonumber
        \sum\limits_{i \in N} d(i,a) &= \sum\limits_{i \in N_1} d(i, a) + \sum\limits_{i \in N_2} d(i, a)\\
        \nonumber
        &\leq \sum\limits_{i \in N_1} d(i, a) + \sum\limits_{i \in N_2} d(i,b) + \frac{\beta}{1-\beta} \sum\limits_{i \in N_1}\left( d(i,b) + d(i,a) \right) & \left(\text{Inequality \eqref{bounds on dist_{N_2}}}\right)\\
        \nonumber
        &\leq D_{\ell} \sum\limits_{i \in N_1} d(i, b) + 
         \sum\limits_{i \in N_2} d(i, b)
         + \frac{\beta}{1-\beta}\left(1 + D_{\ell}\right) \sum\limits_{i \in N_1} d(i, b) & \left(\text{Inequality \eqref{dist on N_1}}\right)\\
        \nonumber
        &\leq \left(D_{\ell} + \frac{\beta}{1 - \beta}\left(1 + D_{\ell}\right)\right) \sum\limits_{i \in N_1} d(i, b) + \left(D_{\ell} + \frac{\beta}{1 - \beta}\left(1 + D_{\ell}\right)\right) \sum\limits_{i \in N_2} d(i, b) & \left(D_{\ell} \geq 1\right)\\
        \nonumber 
        &= \left(D_{\ell} + \frac{\beta}{1 - \beta}\left(1 + D_{\ell}\right)\right) \sum\limits_{i \in N} d(i, b).
        \end{align}
Hence, the distortion of $a$ is at most $D_{\ell} + \frac{\beta}{1 - \beta}\left(1 + D_{\ell}\right)$, as desired. 
\end{proof}


\section{Distortion in Line Metric}
\label{sec:line}
To create a complete picture of how intensities influence metric distortion, in this section we consider the line metric, which is an important and well-studied special case of the general metric, with the goal of providing more refined bounds. 

We start by considering the special case of the elections with two alternatives. We establish a lower bound on the distortion in this setting assuming mandatory elicitation of the intensities, which naturally also serves as a lower bound for any election in the line metric under the same assumption. Then we analyze the structure of the elections in this setting and based on that present a rule which we conjecture to have optimal distortion.
An illustration of this bound is provided in Figure~\ref{illustrations-line:lb}.

\begin{theorem}
    \label{thm:metric_lb}
In the $1$-Euclidean space with two alternatives, the metric distortion of any voting rule $f$ using ranked preferences with intensities ballot format with parameter $\alpha \in (0, 1]$ and assuming mandatory elicitation of the intensities is at least $\max\left( \frac{3 - \alpha}{1 + \alpha}, 2\alpha + 1 \right).$
\end{theorem}

\begin{proof}
We present two lower bounds using two instances that are very similar. Consider elections $\elc =(\ags, \alts, \prefp)$ and $\elc' =(\ags, \alts, \prefp')$ where $\ags = \{1, 2\}$ and $\alts = \{a_1, a_2\}$. The preferences of the agents are
\[
\pref_{i}  =
\begin{cases}
        a_1 \succ a_2 & i = 1\\
        a_2 \succ a_1 & i = 2\\
\end{cases}, \text{ and }
\pref'_{i}  =
\begin{cases}
        a_1 \;\ggc\; a_2 & i = 1\\
        a_2 \;\ggc\; a_1 & i = 2\\
\end{cases}.
\]
Let $f$ be a deterministic rule and w.l.o.g. assume that $a_1 = f(\prefp)$. Now consider metric space $d_1$ as illustrated in \Cref{fig:line_lb_2}. In this metric, the distance between $a_1$ and $a_2$ is $\frac{2}{\alpha} + 2$, and the distances between agents and alternatives are as follows:

\begin{align*}
     &d_1(1, a_1) = \frac{1}{\alpha} + 1, d_1(2, a_1) =\frac{2}{\alpha} + 2 + \frac{2\alpha + 2}{1 - \alpha},\\
     & d_1(1, a_2) = \frac{1}{\alpha} + 1, d_1(2, a_2) = \frac{2\alpha + 2}{1 - \alpha}\cdot
\end{align*}

We have: 
\begin{align*}
\dist_\alpha(a_1, \elc) &\ge \frac{\scsd{d_1}(a_1)}{\scsd{d_1}(a_2)} = \frac{d_1(1, a_1) + d_1(2, a_1)}{d_1(1, a_2) + d_1(2, a_2)}   \\
    &= \frac{\frac{1}{\alpha} + 1 + \frac{2}{\alpha} + 2 + \frac{2\alpha + 2}{1 - \alpha}}{\frac{1}{\alpha} + 1 + \frac{2\alpha + 2}{1 - \alpha}} = \frac{\frac{3}{\alpha^2} + \frac{2}{\alpha} - 1}{\frac{1}{\alpha^2} + \frac{2}{\alpha} + 1}= \frac{3 - \alpha}{1 + \alpha} \cdot
    \end{align*}

For the second election, again w.l.o.g (due to symmetry), assume that $f(\prefp') = a_1$ and consider metric $d_2$ where the distance between $a_1$ and $a_2$ is $\frac{1}{\alpha} + 1$. The first agent lies between $a_1$ and $a_2$ and the second agent is in the same location as $a_2$. The distances between agents and alternatives are as follows:
$$
    d_2(1, a_1) = 1, d_2(2, a_1) = 1+\frac{1}{\alpha}, d_2(1, a_2) = \frac{1}{\alpha}, d_2(2, a_2) = 0.
$$
We have:
\begin{align*}
\dist_\alpha(a_1, \elc') &\ge \frac{\scsd{d_2}(a_1)}{\scsd{d_2}(a_2)} = \frac{d_2(1, a_1) + d_2(2, a_1)}{d_2(1, a_2) + d_2(2, a_2)}  = \frac{ 1 + 1 + \frac{1}{\alpha}}{\frac{1}{\alpha} +0} = 2\alpha + 1 \cdot
    \end{align*}

\begin{figure}[t]
\centering
\begin{minipage}{0.49\textwidth}
  \centering
  \includegraphics[width=0.95\textwidth]{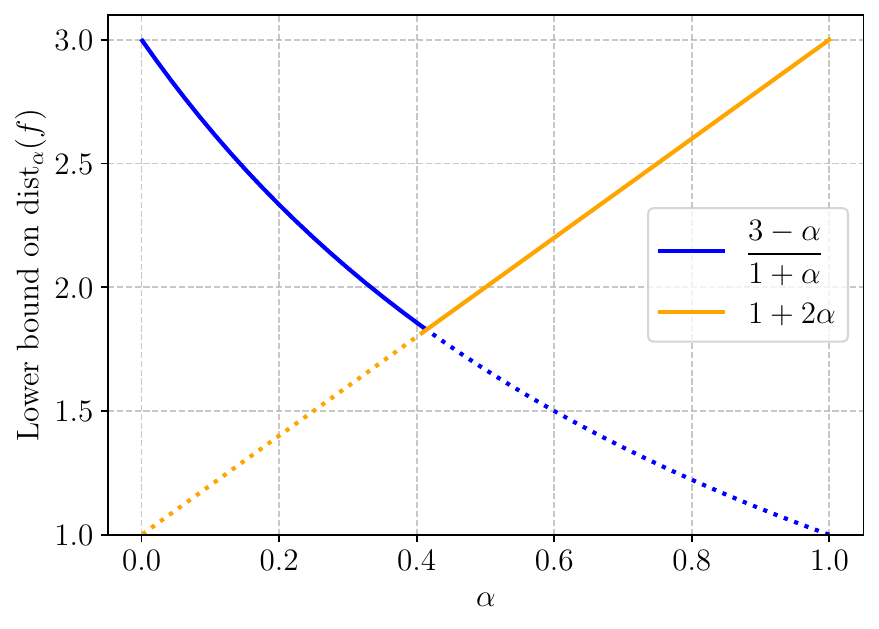}
   \caption{An illustration of the lower bounds established in \Cref{thm:metric_lb}.}
   \label{illustrations-line:lb}
\end{minipage}
\hfill
\begin{minipage}[t]{0.49\textwidth}

\tikzset{every picture/.style={line width=0.5pt}} 
\centering
\vspace{-26mm}
\resizebox{.9\textwidth}{!}{
\begin{tikzpicture}[font=\Large,x=0.75pt,y=0.75pt,yscale=-1,xscale=1]

\draw    (80,130) -- (560,128) ;
\draw    (100,120) -- (100,140) ;
\draw    (560,120) -- (560,140) ;
\draw    (248,120) -- (248,140) ;
\draw   (248.6,112.6) .. controls (248.62,107.93) and (246.3,105.59) .. (241.63,105.57) -- (186,105.34) .. controls (179.33,105.31) and (176.01,102.97) .. (176.03,98.3) .. controls (176.01,102.97) and (172.67,105.28) .. (166,105.25)(169,105.27) -- (110.36,105.02) .. controls (105.69,105) and (103.35,107.32) .. (103.33,111.99) ;
\draw   (393.1,109.2) .. controls (393.09,104.53) and (390.75,102.21) .. (386.08,102.22) -- (332.33,102.37) .. controls (325.66,102.39) and (322.32,100.07) .. (322.31,95.4) .. controls (322.32,100.07) and (319,102.41) .. (312.33,102.43)(315.33,102.42) -- (258.58,102.58) .. controls (253.91,102.59) and (251.59,104.93) .. (251.6,109.6) ;
\draw  [dash pattern={on 4.5pt off 4.5pt}]  (248,188) -- (248,150) ;
\draw [shift={(248,148)}, rotate = 90] [color={rgb, 255:red, 0; green, 0; blue, 0 }  ][line width=0.75]    (10.93,-3.29) .. controls (6.95,-1.4) and (3.31,-0.3) .. (0,0) .. controls (3.31,0.3) and (6.95,1.4) .. (10.93,3.29)   ;
\draw  [dash pattern={on 4.5pt off 4.5pt}]  (560,188) -- (560,150) ;
\draw [shift={(560,148)}, rotate = 90] [color={rgb, 255:red, 0; green, 0; blue, 0 }  ][line width=0.75]    (10.93,-3.29) .. controls (6.95,-1.4) and (3.31,-0.3) .. (0,0) .. controls (3.31,0.3) and (6.95,1.4) .. (10.93,3.29)   ;
\draw    (392,120) -- (392,140) ;
\draw   (560.1,107.2) .. controls (560.09,102.53) and (557.75,100.21) .. (553.08,100.22) -- (501.04,100.34) .. controls (494.37,100.36) and (491.04,98.04) .. (491.03,93.37) .. controls (491.04,98.04) and (487.71,100.38) .. (481.04,100.39)(484.04,100.39) -- (408.58,100.57) .. controls (403.91,100.58) and (401.59,102.92) .. (401.6,107.59) ;

\draw (460,50.4) node [anchor=north west][inner sep=0.75pt]    {$\frac{2\alpha \ +\ 2}{1\ -\ \alpha }$};
\draw (153,55.4) node [anchor=north west][inner sep=0.75pt]    {$\frac{1}{\alpha } +\ 1$};
\draw (243,194.4) node [anchor=north west][inner sep=0.75pt]    {$1$};
\draw (555,194.4) node [anchor=north west][inner sep=0.75pt]    {$2$};
\draw (83,128.4) node [anchor=north west][inner sep=0.75pt]    {$a_{1}$};
\draw (374,134.4) node [anchor=north west][inner sep=0.75pt]    {$a_{2}$};
\draw (71,52.4) node [anchor=north west][inner sep=0.75pt]    {$d_{1} :$};
\draw (301,54.4) node [anchor=north west][inner sep=0.75pt]    {$\frac{1}{\alpha } +\ 1$};
\end{tikzpicture}}

\resizebox{.9\textwidth}{!}{
\begin{tikzpicture}[font=\Large,x=0.75pt,y=0.75pt,yscale=-1,xscale=1]

\draw    (80,130) -- (560,128) ;
\draw    (100,120) -- (100,140) ;
\draw    (560,120) -- (560,140) ;
\draw    (248,120) -- (248,140) ;
\draw   (248.6,112.6) .. controls (248.62,107.93) and (246.3,105.59) .. (241.63,105.57) -- (186,105.34) .. controls (179.33,105.31) and (176.01,102.97) .. (176.03,98.3) .. controls (176.01,102.97) and (172.67,105.28) .. (166,105.25)(169,105.27) -- (110.36,105.02) .. controls (105.69,105) and (103.35,107.32) .. (103.33,111.99) ;
\draw   (559.6,111.6) .. controls (559.63,106.93) and (557.32,104.58) .. (552.65,104.55) -- (415.64,103.67) .. controls (408.98,103.63) and (405.66,101.28) .. (405.69,96.61) .. controls (405.66,101.28) and (402.32,103.59) .. (395.65,103.54)(398.65,103.56) -- (258.64,102.66) .. controls (253.98,102.63) and (251.63,104.94) .. (251.6,109.61) ;
\draw  [dash pattern={on 4.5pt off 4.5pt}]  (248,188) -- (248,150) ;
\draw [shift={(248,148)}, rotate = 90] [color={rgb, 255:red, 0; green, 0; blue, 0 }  ][line width=0.75]    (10.93,-3.29) .. controls (6.95,-1.4) and (3.31,-0.3) .. (0,0) .. controls (3.31,0.3) and (6.95,1.4) .. (10.93,3.29)   ;
\draw  [dash pattern={on 4.5pt off 4.5pt}]  (560,188) -- (560,150) ;
\draw [shift={(560,148)}, rotate = 90] [color={rgb, 255:red, 0; green, 0; blue, 0 }  ][line width=0.75]    (10.93,-3.29) .. controls (6.95,-1.4) and (3.31,-0.3) .. (0,0) .. controls (3.31,0.3) and (6.95,1.4) .. (10.93,3.29)   ;

\draw (400,50.4) node [anchor=north west][inner sep=0.75pt]    {$\frac{1}{\alpha }$};
\draw (169,66.4) node [anchor=north west][inner sep=0.75pt]    {$1$};
\draw (243,194.4) node [anchor=north west][inner sep=0.75pt]    {$1$};
\draw (555,194.4) node [anchor=north west][inner sep=0.75pt]    {$2$};
\draw (83,128.4) node [anchor=north west][inner sep=0.75pt]    {$a_{1}$};
\draw (562,131.4) node [anchor=north west][inner sep=0.75pt]    {$a_{2}$};
\draw (71,52.4) node [anchor=north west][inner sep=0.75pt]    {$d_{2} :$};
\end{tikzpicture}}
\caption{Metrics used in the proof of \Cref{thm:metric_lb}.}
\label{fig:line_lb_2}
\end{minipage}
\end{figure}

One can easily check that $\prefp \scns d_1$ and $\prefp' \scns d_2$. By definition, the distortion of $f$ is at least the maximum of its distortion in $\elc$ and $\elc'$, and
hence, $\dist_{\alpha}(f)$ is at least $\max\left( \frac{3 - \alpha}{1 + \alpha}, 2\alpha + 1 \right)$. \qedhere
\end{proof}

In order to design a rule with good distortion in this restricted setting, let us first take a look at the structure of the instances. In an election with two alternatives $\{a_1, a_2\}$, agents have four possible preferences: $a_1 \succ a_2$,  $a_2 \succ a_1$, $a_1 \ggc a_2$, and $a_2 \ggc a_1$. Let $N_{a_1 \succ a_2}$,  $N_{a_2 \succ a_1}$, $N_{a_1 \ggc a_2}$, and $N_{a_2 \ggc a_1}$ indicate the set of agents with these preferences respectively. 
Our goal is to decide between $a_1$ and $a_2$ which one has a lower distortion for a given preference profile $\prefp$. To make this decision, we think of metrics that are consistent with $\prefp$ and maximize the distortion of a specific alternative. For alternative $a_1$, consider metrics $d_1$ and $d_2$ as illustrated in \Cref{fig:metric_worst}. In the following lemma, we prove that these metrics have the maximum distortion for $a_1$.

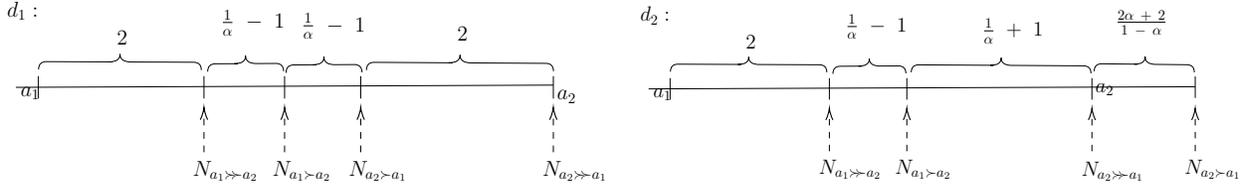
\begin{figure}[t]	
\centering
\begin{minipage}[t]{0.49\textwidth}
\resizebox{\textwidth}{!}{
\begin{tikzpicture}[font=\Large,x=0.75pt,y=0.75pt,yscale=-1,xscale=1]

\draw    (80,130) -- (560,128) ;
\draw    (100,120) -- (100,140) ;
\draw    (560,120) -- (560,140) ;
\draw    (248,120) -- (248,140) ;
\draw   (245.6,114.6) .. controls (245.62,109.93) and (243.3,107.59) .. (238.63,107.57) -- (183,107.34) .. controls (176.33,107.31) and (173.01,104.97) .. (173.03,100.3) .. controls (173.01,104.97) and (169.67,107.28) .. (163,107.25)(166,107.27) -- (107.36,107.02) .. controls (102.69,107) and (100.35,109.32) .. (100.33,113.99) ;
\draw   (560.1,115.6) .. controls (560.13,110.93) and (557.81,108.59) .. (553.14,108.56) -- (486.64,108.16) .. controls (479.97,108.12) and (476.65,105.77) .. (476.68,101.1) .. controls (476.65,105.77) and (473.31,108.08) .. (466.64,108.05)(469.64,108.06) -- (400.14,107.65) .. controls (395.47,107.62) and (393.13,109.94) .. (393.1,114.61) ;
\draw  [dash pattern={on 4.5pt off 4.5pt}]  (248,188) -- (248,150) ;
\draw [shift={(248,148)}, rotate = 90] [color={rgb, 255:red, 0; green, 0; blue, 0 }  ][line width=0.75]    (10.93,-3.29) .. controls (6.95,-1.4) and (3.31,-0.3) .. (0,0) .. controls (3.31,0.3) and (6.95,1.4) .. (10.93,3.29)   ;
\draw  [dash pattern={on 4.5pt off 4.5pt}]  (560,188) -- (560,150) ;
\draw [shift={(560,148)}, rotate = 90] [color={rgb, 255:red, 0; green, 0; blue, 0 }  ][line width=0.75]    (10.93,-3.29) .. controls (6.95,-1.4) and (3.31,-0.3) .. (0,0) .. controls (3.31,0.3) and (6.95,1.4) .. (10.93,3.29)   ;
\draw  [dash pattern={on 4.5pt off 4.5pt}]  (320,188) -- (320,150) ;
\draw [shift={(320,148)}, rotate = 90] [color={rgb, 255:red, 0; green, 0; blue, 0 }  ][line width=0.75]    (10.93,-3.29) .. controls (6.95,-1.4) and (3.31,-0.3) .. (0,0) .. controls (3.31,0.3) and (6.95,1.4) .. (10.93,3.29)   ;
\draw    (320,120) -- (320,129.6) -- (320,140) ;
\draw   (319.1,115.6) .. controls (319.07,110.93) and (316.73,108.61) .. (312.06,108.64) -- (295.17,108.74) .. controls (288.51,108.78) and (285.16,106.47) .. (285.13,101.8) .. controls (285.16,106.47) and (281.85,108.82) .. (275.18,108.86)(278.18,108.84) -- (258.29,108.96) .. controls (253.62,108.99) and (251.3,111.33) .. (251.33,116) ;
\draw  [dash pattern={on 4.5pt off 4.5pt}]  (388,188) -- (388,150) ;
\draw [shift={(388,148)}, rotate = 90] [color={rgb, 255:red, 0; green, 0; blue, 0 }  ][line width=0.75]    (10.93,-3.29) .. controls (6.95,-1.4) and (3.31,-0.3) .. (0,0) .. controls (3.31,0.3) and (6.95,1.4) .. (10.93,3.29)   ;
\draw    (388,120) -- (388,132) -- (388,140) ;
\draw   (389.1,116.6) .. controls (389.07,111.93) and (386.73,109.61) .. (382.06,109.64) -- (365.17,109.74) .. controls (358.51,109.78) and (355.16,107.47) .. (355.13,102.8) .. controls (355.16,107.47) and (351.85,109.82) .. (345.18,109.86)(348.18,109.84) -- (328.29,109.96) .. controls (323.62,109.99) and (321.3,112.33) .. (321.33,117) ;

\draw (473,74.4) node [anchor=north west][inner sep=0.75pt]    {$2$};
\draw (169,78.4) node [anchor=north west][inner sep=0.75pt]    {$2$};
\draw (237,194.4) node [anchor=north west][inner sep=0.75pt]    {$N_{a_1 \ggc a_2}$};
\draw (549,194.4) node [anchor=north west][inner sep=0.75pt]    {$N_{a_2 \ggc a_1}$};
\draw (83,128.4) node [anchor=north west][inner sep=0.75pt]    {$a_{1}$};
\draw (562,131.4) node [anchor=north west][inner sep=0.75pt]    {$a_{2}$};
\draw (71,52.4) node [anchor=north west][inner sep=0.75pt]    {$d_{1} :$};
\draw (308,194.4) node [anchor=north west][inner sep=0.75pt]    {$N_{a_1 \succ a_2}$};
\draw (262,59.4) node [anchor=north west][inner sep=0.75pt]    {$\frac{1}{\alpha } \ -\ 1$};
\draw (376,194.4) node [anchor=north west][inner sep=0.75pt]    {$N_{a_2 \succ a_1}$};
\draw (334,62.4) node [anchor=north west][inner sep=0.75pt]    {$\frac{1}{\alpha } \ -\ 1$};

\end{tikzpicture}}
\end{minipage}
\hfill
\begin{minipage}[t]{0.49\textwidth}
\resizebox{\textwidth}{!}{
\begin{tikzpicture}[font=\Large,x=0.75pt,y=0.75pt,yscale=-1,xscale=1]

\draw    (80,130) -- (588,128) ;
\draw    (100,120) -- (100,140) ;
\draw    (492,120) -- (492,140) ;
\draw    (248,120) -- (248,140) ;
\draw   (245.6,114.6) .. controls (245.62,109.93) and (243.3,107.59) .. (238.63,107.57) -- (183,107.34) .. controls (176.33,107.31) and (173.01,104.97) .. (173.03,100.3) .. controls (173.01,104.97) and (169.67,107.28) .. (163,107.25)(166,107.27) -- (107.36,107.02) .. controls (102.69,107) and (100.35,109.32) .. (100.33,113.99) ;
\draw   (491.1,115.6) .. controls (491.13,110.93) and (488.81,108.59) .. (484.14,108.56) -- (417.64,108.16) .. controls (410.97,108.12) and (407.65,105.77) .. (407.68,101.1) .. controls (407.65,105.77) and (404.31,108.08) .. (397.64,108.05)(400.64,108.06) -- (331.14,107.65) .. controls (326.47,107.62) and (324.13,109.94) .. (324.1,114.61) ;
\draw  [dash pattern={on 4.5pt off 4.5pt}]  (248,188) -- (248,150) ;
\draw [shift={(248,148)}, rotate = 90] [color={rgb, 255:red, 0; green, 0; blue, 0 }  ][line width=0.75]    (10.93,-3.29) .. controls (6.95,-1.4) and (3.31,-0.3) .. (0,0) .. controls (3.31,0.3) and (6.95,1.4) .. (10.93,3.29)   ;
\draw  [dash pattern={on 4.5pt off 4.5pt}]  (492,188) -- (492,150) ;
\draw [shift={(492,148)}, rotate = 90] [color={rgb, 255:red, 0; green, 0; blue, 0 }  ][line width=0.75]    (10.93,-3.29) .. controls (6.95,-1.4) and (3.31,-0.3) .. (0,0) .. controls (3.31,0.3) and (6.95,1.4) .. (10.93,3.29)   ;
\draw  [dash pattern={on 4.5pt off 4.5pt}]  (320,188) -- (320,150) ;
\draw [shift={(320,148)}, rotate = 90] [color={rgb, 255:red, 0; green, 0; blue, 0 }  ][line width=0.75]    (10.93,-3.29) .. controls (6.95,-1.4) and (3.31,-0.3) .. (0,0) .. controls (3.31,0.3) and (6.95,1.4) .. (10.93,3.29)   ;
\draw    (320,120) -- (320,129.6) -- (320,140) ;
\draw   (319.1,115.6) .. controls (319.07,110.93) and (316.73,108.61) .. (312.06,108.64) -- (295.17,108.74) .. controls (288.51,108.78) and (285.16,106.47) .. (285.13,101.8) .. controls (285.16,106.47) and (281.85,108.82) .. (275.18,108.86)(278.18,108.84) -- (258.29,108.96) .. controls (253.62,108.99) and (251.3,111.33) .. (251.33,116) ;
\draw  [dash pattern={on 4.5pt off 4.5pt}]  (588,188) -- (588,150) ;
\draw [shift={(588,148)}, rotate = 90] [color={rgb, 255:red, 0; green, 0; blue, 0 }  ][line width=0.75]    (10.93,-3.29) .. controls (6.95,-1.4) and (3.31,-0.3) .. (0,0) .. controls (3.31,0.3) and (6.95,1.4) .. (10.93,3.29)   ;
\draw    (588,118.4) -- (588,128) -- (588,138.4) ;
\draw   (587.1,114.6) .. controls (587.15,109.93) and (584.85,107.57) .. (580.18,107.52) -- (551.42,107.21) .. controls (544.75,107.14) and (541.44,104.78) .. (541.49,100.11) .. controls (541.44,104.78) and (538.09,107.07) .. (531.42,107)(534.42,107.03) -- (501.18,106.67) .. controls (496.51,106.62) and (494.15,108.92) .. (494.1,113.59) ;

\draw (389,63.4) node [anchor=north west][inner sep=0.75pt]    {$\frac{1}{\alpha } \ +\ 1$};
\draw (169,78.4) node [anchor=north west][inner sep=0.75pt]    {$2$};
\draw (237,194.4) node [anchor=north west][inner sep=0.75pt]    {$N_{a_1 \ggc a_2}$};
\draw (481,196.4) node [anchor=north west][inner sep=0.75pt]    {$N_{a_2 \ggc a_1}$};
\draw (83,128.4) node [anchor=north west][inner sep=0.75pt]    {$a_{1}$};
\draw (494,123.4) node [anchor=north west][inner sep=0.75pt]    {$a_{2}$};
\draw (71,52.4) node [anchor=north west][inner sep=0.75pt]    {$d_{2} :$};
\draw (308,194.4) node [anchor=north west][inner sep=0.75pt]    {$N_{a_1 \succ a_2}$};
\draw (262,59.4) node [anchor=north west][inner sep=0.75pt]    {$\frac{1}{\alpha } \ -\ 1$};
\draw (577,194.4) node [anchor=north west][inner sep=0.75pt]    {$N_{a_2 \succ a_1}$};
\draw (513,55.4) node [anchor=north west][inner sep=0.75pt]    {$\frac{2\alpha \ +\ 2}{1\ -\ \alpha }$};
\end{tikzpicture}}
\end{minipage}
\vspace{-1cm}
\caption{Metrics that provide worst possible distortion.}
\label{fig:metric_worst}
\end{figure}

\begin{lemma}
\label{lem:line-tech}
Consider an election $\elc$ with two alternatives $\alts = \{a_1, a_2\}$, and metric spaces $d_1$ and $d_2$ as illustrated in \Cref{fig:metric_worst}, we have:
    $$\dist_\alpha(a_1, \elc) = \max \left(\frac{\scsd{d_1}(a_1)}{\scsd{d_1}(a_2)}, \frac{\scsd{d_2}(a_1)}{\scsd{d_2}(a_2)}, 1 \right).$$
\end{lemma}

\begin{proof}
     Assume that $d$ is a line metric that maximizes the ratio between the social cost of $a_1$ and that of $a_2$. For each agent we find a location for her in $d$ that maximizes her contribution to the distortion using the fact that $d(i, a_1)$ and $d(i, a_2)$ appear in the numerator and denominator of $\frac{\scsd{d}(a_1)}{\scsd{d}(a_2)}$ respectively. We have to consider four types of agents. First let us put $a_1$ to the left and $a_2$ to the right. For an agent $i \in N_{a_2 \ggc a_1}$,
     if $i$ is not at the same location as $a_2$, it either lies between the two alternatives or is located to the right of $a_2$. If it is in between by moving it towards $a_2$ we can decrease $d(i, a_2)$ and increase $d(i, a_1)$ which makes $\frac{\scsd{d}(a_1)}{\scsd{d}(a_2)}$ larger. In the other case, if we move $i$ closer to $a_2$, denominator and numerator will become equally smaller. Yet, because the distortion of $a_1$ is greater than one, the ratio will become larger. Therefore, in both cases, we are able to make $\frac{\scsd{d}(a_1)}{\scsd{d}(a_2)}$ larger, which shows that $i$ must be at the same location as $a_2$ to make the maximum contribution to the distortion.
     
     For an agent $i \in N_{a_1 \succ a_2}$, assume that $d(i, a_1) < d(i, a_2)$. If $i$ lies between the two alternatives, we can move it closer to the middle of the segment connecting alternatives (farther from $a_1$), and make the numerator larger, and the denominator smaller. This gives us a larger distortion. If $i$ lies to the left of $a_1$, we move it to the middle of the segment between alternatives. By doing so, the denominator will become smaller ($d(i, a_2)$ was at least as large as the length of the segment), and the numerator will probably become smaller, but the value of the numerator will not change as much as the value of denominator will. Therefore, the distortion will become larger, as it was greater than one. Hence, setting $d(i, a_1) = d(i, a_2)$ makes $i$ make the maximum contribution to the distortion.

     For $i \in N_{a_1 \ggc a_2}$, first assume that $i$ lies outside the segment connecting the alternatives. With an argument similar to the previous case, we can easily show that if we move $i$ to a location between the two alternatives where $d(i, a_1) \ge \alpha d(i, a_2)$, we will obtain a larger distortion. Hence, $i$ must be between the two alternatives. Here, the closer $i$ gets to $a_2$, the larger the distortion will be. Due to intensity constraints, $i$ can at most lie at a location where $d(i, a_1) = \alpha d(i, a_2)$, and that is where it makes the maximum contribution to the distortion.
     
     Finally, for $i \in N_{a_2 \succ a_1}$, there are two possible locations. If $i$ lies between the alternatives, with similar arguments, it has to lie at the closest possible location to $a_2$, which, due to intensity constraints, is a location where $\alpha d(i, a_1) = d(i, a_2)$. If $i$ lies outside the segment connecting the alternatives, by move it closer to $a_2$, we can make both numerator and denominator equally smaller, which gives us a larger distortion. That means, it has to be located at the closest possible location to $a_2$, which is a locations where $\alpha d(i, a_1) = d(i, a_2)$ due to intensity constraints. In this case $i$ can lie in two possible locations. We can now conclude that the metric $d$ which maximizes the distortion of $a_1$ is similar to either $d_1$ or $d_2$, based on the location of agents in $N_{a_2 \succ a_1}$, which completes the proof.
\end{proof}

Note that we can define metrics $d'_1$ and $d'_2$ that maximize the distortion of $a_2$ by changing the location of $a_1$ and $a_2$ in $d_1$ and $d_2$ and relabel the agents respectively. This observation gives us the tool to find the alternative with minimum distortion in elections with two alternatives.

We can see that 

\[
\frac{\scsd{d_1}(a_1)}{\scsd{d_1}(a_2)}\!=\!\frac{(\frac{1}{\alpha} + 1)  n_{a_1 \succ a_2} + \frac{2}{\alpha}  n_{a_2 \succ a_1} + 2  n_{a_1 \ggc a_2} + (2 + \frac{2}{\alpha})  n_{a_2 \ggc a_1}}{(\frac{1}{\alpha} + 1)  n_{a_1 \succ a_2} + 2  n_{a_2 \succ a_1} + \frac{2}{\alpha}  n_{a_1 \ggc a_2}}
\]
where $n_X := |N_X|$.
Similarly we have
\[
\frac{\scsd{d_2}\!(a_1\!)}{\scsd{d_2}\!(a_2\!)}\!=\!\frac{(\frac{1}{\alpha} + 1)  n_{a_1 \succ a_2}\!+\!\frac{2 + \frac{2}{\alpha}}{1 - \alpha}  n_{a_2 \succ a_1}\!+\!2  n_{a_1 \ggc a_2}\!+\!(2 + \frac{2}{\alpha})  n_{a_2 \ggc a_1}}{(\frac{1}{\alpha} + 1)  n_{a_1 \succ a_2} + \frac{2\alpha + 2}{1 - \alpha}  n_{a_2 \succ a_1} + \frac{2}{\alpha}  n_3}\cdot
\]
For convenience let
\begin{align*}
D_{\alpha}(n_1, n_2, n_3, n_4)=\max \Bigg(&\frac{(\frac{1}{\alpha} + 1)  n_1 + \frac{2}{\alpha}  n_2 + 2  n_3 + (2 + \frac{2}{\alpha})  n_4}{(\frac{1}{\alpha} + 1)  n_1 + 2  n_2 + \frac{2}{\alpha}  n_3},\\
&\frac{(\frac{1}{\alpha} + 1)  n_1 + \frac{\frac{2}{\alpha}(1 + \frac{1}{\alpha})}{\frac{1}{\alpha} - 1}  n_2 + 2  n_3 + (2 + \frac{2}{\alpha})  n_4}{(\frac{1}{\alpha} + 1)  n_1 + \frac{2(1 + \frac{1}{\alpha})}{\frac{1}{\alpha} - 1}  n_2 + \frac{2}{\alpha}  n_3} \Bigg).
\end{align*}
That means 
$$\dist_\alpha(a_1, \elc) = D_\alpha(n_{a_1 \succ a_2}, n_{a_2 \succ a_1}, n_{a_1 \ggc a_2}, n_{a_2 \ggc a_1}).$$
With the same argument we can show that 
$$\dist_\alpha(a_2, \elc) = D_\alpha(n_{a_2 \succ a_1}, n_{a_1 \succ a_2}, n_{a_2 \ggc a_1}, n_{a_1 \ggc a_2}).$$
That means a voting rule can simply compute $D_\alpha(n_{a_1 \succ a_2},  n_{a_2 \succ a_1}, n_{a_1 \ggc a_2}, n_{a_2 \ggc a_1})$ and $D_\alpha(n_{a_2 \succ a_1},\allowbreak n_{a_1 \succ a_2}, n_{a_2 \ggc a_1}, n_{a_1 \ggc a_2})$ and select $a_1$ if the former is smaller and $a_2$ otherwise. For a given $\alpha$ we call this rule $f^{\mathsf{tal}}_\alpha$ since it performs well when we have two agents on a line.

\begin{theorem}
    Given an election $\elc = (N, \{a_1, a_2\}, \prefp)$ in the $1$-Euclidean space, and parameter $\alpha \in (0, 1]$ we have 

\[
\dist_\alpha\left(f^{\mathsf{tal}}_\alpha(\prefp), \elc\right) = 
\min\Bigg\{
\begin{aligned}
&D_\alpha(n_{a_1 \succ a_2},\; n_{a_2 \succ a_1},\; n_{a_1 \ggc a_2},\; n_{a_2 \ggc a_1}), \\
&D_\alpha(n_{a_2 \succ a_1},\; n_{a_1 \succ a_2},\; n_{a_2 \ggc a_1},\; n_{a_1 \ggc a_2})
\end{aligned}
\Bigg\}
\]

\label{thm:metric_ub}
\end{theorem}

Note that this rule outputs the instance-optimal alternative. Furthermore, using \Cref{thm:metric_ub} we can give the following bound on the distortion of this rule. 
$
\dist_\alpha (f^{\mathsf{tal}}_\alpha) = \max_{n_1, n_2, n_3, n_4} (\min (D_{\alpha}(n_2,\allowbreak  n_1, n_4, n_3), D_{\alpha}(n_1, n_2, n_3, n_4) ))
$
where in this equation distortion is defined over elections with two alternatives and in the $1$-Euclidean metric space. Although we were unable to algebraically simplify the worst-case distortion of this rule, we conjecture that in the worst case the distortion of this rule matches the lower bound in \Cref{thm:metric_lb}.

\begin{conjecture}
Considering elections with two alternatives and assuming that the metric is limited to the $1$-Euclidean space, we have
    $$\dist_\alpha (f^{\mathsf{tal}}_\alpha) = \max\left( \frac{3 - \alpha}{1 + \alpha}, 2\alpha + 1 \right)\cdot$$
\end{conjecture}

Although we could not prove this conjecture algebraically, our simulations support it. Actually, we have computed the value of $\min\!\left(D_{\alpha}(n_2, n_1, n_4, n_3), D_{\alpha}(n_1, n_2, n_3, n_4) \right) $ for all possible values of $n_1, n_2, n_3, n_4$ that sum up to 1000.

We conclude this section by providing a general lower bound on the distortion of voting rules in the line metric.
Recall that in \Cref{lem:mand_lb_left,lem:mand_lb_right}, we established a lower bound on the distortion, but that lower bound does not imply any lower bound for the line metric since the metric that we used in the proof was more general. That means we need a new lower bound for this setting.

\begin{restatable}{theorem}{linelowewrbound}\label{general_lower_bound for line}
In the $1$-Euclidean space, the metric distortion of any voting rule $f$ using ranked preferences with intensities ballot format with parameter $\alpha \in (0, 1]$ and assuming mandatory elicitation of the intensities is at least 
\[
\frac{1 + 3\alpha^{-\lfloor\frac{m}{2}\rfloor}}{3 + \alpha^{-\lfloor\frac{m}{2}\rfloor}}\cdot
\]
\end{restatable}

\begin{proof}
First, we assume that $m$ is even. Consider an election \(\elc = (\ags, \alts, \prefp)\) where \(\ags = \{i_1, i_2\}\) and \(\alts = \{a_1, a_2, \ldots, a_{\frac{m}{2}}, b_1, b_2, \ldots, b_{\frac{m}{2}}\}\). The preferences of agents are as follows:
\begin{equation}
\nonumber \pref_{i_l}  =
\begin{cases}
        a_1 \succ a_2 \succ \ldots \succ a_{\frac{m}{2}} \succ b_1 \succ b_2 \succ \ldots \succ b_{\frac{m}{2}}  &\; l = 1\\
        b_1 \succ b_2 \succ \ldots \succ b_{\frac{m}{2}} \succ a_1 \succ a_2 \succ \ldots \succ a_{\frac{m}{2}} &\; l = 2
\end{cases}
\end{equation}
Let $f$ be a deterministic social choice rule and define $c := f(\prefp)$. Note that $c$ is either $a_l$ for some $1 \leq l \leq m/2$ or $b_j$ for some $1 \leq j \leq m/2$. Due to the symmetry, without loss of generality, assume that $c$ is $a_l$.

 Consider the following metric \( d \): candidates \( a_1,a_2,\ldots,a_{m/2}\) are located at a single location. agent \( i_1 \) is placed at a distance of \(\left(\frac{1}{\alpha^{\frac{m}{2}}} + 1\right)\) to the right of \( a_1 \), and \( i_2 \) is positioned at a distance of \(\left(\frac{1}{\alpha^{\frac{m}{2}}} - 1\right)\) to the right of \( i_1 \). Candidates \( b_1, b_2, \ldots, b_{\frac{m}{2}} \) are arranged from left to right, starting at the right of \(i_2\), such that the distance from \( b_1 \) to \( i_2 \) is 2, and the distance between \( b_j \) and \( b_{j+1} \) is \(\left(\frac{2}{\alpha^j}\right) - 2\) for \( 1 \leq j \leq \frac{m}{2} - 1 \).
See Figure \ref{fig:metric d} for an illustration of this metric.

We need to check that \(\prefp \scns d\). First, it is obvious that \(d(i_2, b_j) = \frac{2}{{\alpha}^{(j - 1)}}\) and \(d(i_2, a_j) = \frac{2}{{\alpha}^{\frac{m}{2}}}\) for all \(1 \leq j \leq \frac{m}{2}\). Therefore, \(i_2\) satisfies the intensity constraints.
For \(d(i_1, b_j)\), if $1 \leq j \leq \frac{m}{2}-1$, we have:

\begin{align*}
    \frac{1}{\alpha} d(i_1, b_j) =&\; \left(\frac{1}{\alpha}\right)\left(\frac{1}{{\alpha}^{\frac{m}{2}}} - 1 + \frac{2}{{\alpha}^{(j-1)}}\right) = \frac{1}{\alpha^{\frac{m}{2} + 1}} - \frac{1}{\alpha} + \frac{2}{{\alpha}^{j}} \\
    =& \left(\frac{1}{\alpha}\right)\left(\frac{1}{\alpha^{\frac{m}{2}}} - 1\right) + \frac{2}{{\alpha}^{j}} \ge \frac{1}{\alpha^{\frac{m}{2}}} - 1 + \frac{2}{{\alpha}^{j}} = d(i_1, b_{j + 1}).
\end{align*}

For $d(i_1, b_{\frac{m}{2}})$, we have:
\begin{align*}
\frac{1}{\alpha} d(i_1, b_{\frac{m}{2}}) =&\;\left(\frac{1}{\alpha}\right)\left(\frac{1}{{\alpha}^{\frac{m}{2}}} - 1 + \frac{2}{{\alpha}^{(\frac{m}{2} - 1)}}\right) =\; \left(\frac{1}{{\alpha}^{(\frac{m}{2} + 1)}} - \frac{1}{\alpha} + \frac{2}{{\alpha}^{\frac{m}{2}}}\right) \geq\; \frac{2}{{\alpha}^{\frac{m}{2}}} \geq\; \frac{1}{{\alpha}^{\frac{m}{2}}} + 1 =\; d(i_1, a_1).
\end{align*}
Also, for $d(i_1, a_1)$, if $1 \leq j\leq \frac{m}{2}$, we have: $
d(i_1, a_1) = d(i_1, a_j).$
Hence, $i_1$ also satisfies the intensity constraints. Now, note that $\scs(a_1) \leq \scs(a_l)$, since: 
\begin{align*}
\scs(a_1) =&\; d(i_1, a_1) + d(i_2, a_1) \\\leq &\; d(i_1, a_l) + d(i_2, a_l)  & \text{(since  $a_1 \succ_{i_1} a_l,\;\; a_1 \succ_{i_2} a_l$)} = \scs(a_l).
\end{align*}
As a consequence, the distortion of \(f\) is at least:
\begin{align*}
\frac{\scs(a_l)}{\scs(b_1)} \geq \frac{\scs(a_l)}{\scs(b_1)} = \frac{d(i_1, a_1) + d(i_2, a_1)}{d(i_1, b_1) + d(i_2, b_1)} = 
\frac{1 + \frac{3}{\alpha^{\frac{m}{2}}}}{3 + \frac{1}{\alpha^{\frac{m}{2}}}} = \frac{1 + 3\alpha^{-\frac{m}{2}}}{3 + \alpha^{-\frac{m}{2}}}\cdot
\end{align*}
Now, assume that $m$ is odd. Consider an election \(\elc = (\ags, \alts, \prefp)\) where \(\ags = \{i_1, i_2\}\) and \(\alts = \{a_1, a_2, \ldots, a_{\frac{m-1}{2}},\allowbreak b_1, b_2, \ldots, b_{\frac{m-1}{2}}, c\}\). The preferences of agents are as follows:
\begin{equation}
\nonumber \pref_{i_l}  =
\begin{cases}
        a_1 \succ a_2 \succ \ldots \succ a_{\frac{m-1}{2}} \succ b_1 \succ b_2 \succ \ldots \succ b_{\frac{m-1}{2}} \;\ggc\; c &\; l = 1\\
        b_1 \succ b_2 \succ \ldots \succ b_{\frac{m-1}{2}} \succ a_1 \succ a_2 \succ \ldots \succ a_{\frac{m-1}{2}} \;\ggc\; c &\; l = 2.
\end{cases}
\end{equation}
We drop $c$ from the election since $c$ is the least preferred alternative for each agent. Now, we can apply the metric $d$ we found for the case where $m$ is even to the remaining alternatives and agents, and obtain the following lower bound on the distortion:
\[
\frac{1 + 3\alpha^{-\frac{m-1}{2}}}{3 + \alpha^{-\frac{m-1}{2}}}\cdot
\]
To complete this argument, we have to determine the location of $c$ in the line metric $d$. We put $c$ at a far distance from all other agents and alternatives. It is now easy to check that this placement satisfies the intensity and preference constraints. 
\end{proof}
\begin{figure}[t]	

\centering
\resizebox{.5\textwidth}{!}{

\begin{tikzpicture}[font=\huge,x=0.75pt,y=0.75pt,yscale=-1,xscale=1]

\draw    (100,130) -- (319.6,130) -- (640,130) ;
\draw    (100,120) -- (100,140) ;
\draw    (200,120) -- (200,130) -- (200,140) ;
\draw    (640,120) -- (640,140) ;
\draw    (380,120) -- (380,140) ;
\draw    (270,120) -- (270,127.6) -- (270,140) ;
\draw   (199.6,114.6) .. controls (199.6,109.93) and (197.27,107.6) .. (192.6,107.6) -- (159.6,107.6) .. controls (152.93,107.6) and (149.6,105.27) .. (149.6,100.6) .. controls (149.6,105.27) and (146.27,107.6) .. (139.6,107.6)(142.6,107.6) -- (106.6,107.6) .. controls (101.93,107.6) and (99.6,109.93) .. (99.6,114.6) ;
\draw   (448.6,110.85) .. controls (448.55,106.18) and (446.2,103.87) .. (441.53,103.92) -- (425.14,104.1) .. controls (418.47,104.17) and (415.11,101.87) .. (415.06,97.2) .. controls (415.11,101.87) and (411.81,104.23) .. (405.14,104.3)(408.14,104.27) -- (388.74,104.48) .. controls (384.07,104.53) and (381.77,106.88) .. (381.82,111.55) ;
\draw   (269.6,113.6) .. controls (269.53,108.93) and (267.17,106.63) .. (262.5,106.7) -- (245.34,106.93) .. controls (238.67,107.02) and (235.31,104.74) .. (235.24,100.07) .. controls (235.31,104.74) and (232.01,107.12) .. (225.34,107.21)(228.34,107.17) -- (208.18,107.45) .. controls (203.51,107.51) and (201.21,109.87) .. (201.27,114.54) ;
\draw   (321.6,112.6) .. controls (321.56,107.93) and (319.21,105.62) .. (314.54,105.66) -- (306.41,105.72) .. controls (299.74,105.77) and (296.39,103.47) .. (296.35,98.8) .. controls (296.39,103.47) and (293.08,105.83) .. (286.41,105.88)(289.41,105.86) -- (278.28,105.95) .. controls (273.61,105.98) and (271.3,108.33) .. (271.33,113) ;
\draw  [dash pattern={on 4.5pt off 4.5pt}]  (200,190) -- (200,159.6) -- (200,152) ;
\draw [shift={(200,150)}, rotate = 90] [color={rgb, 255:red, 0; green, 0; blue, 0 }  ][line width=0.75]    (10.93,-3.29) .. controls (6.95,-1.4) and (3.31,-0.3) .. (0,0) .. controls (3.31,0.3) and (6.95,1.4) .. (10.93,3.29)   ;
\draw  [dash pattern={on 4.5pt off 4.5pt}]  (270,190) -- (270,152) ;
\draw [shift={(270,150)}, rotate = 90] [color={rgb, 255:red, 0; green, 0; blue, 0 }  ][line width=0.75]    (10.93,-3.29) .. controls (6.95,-1.4) and (3.31,-0.3) .. (0,0) .. controls (3.31,0.3) and (6.95,1.4) .. (10.93,3.29)   ;
\draw  [dash pattern={on 4.5pt off 4.5pt}]  (320,190) -- (320,152) ;
\draw [shift={(320,150)}, rotate = 90] [color={rgb, 255:red, 0; green, 0; blue, 0 }  ][line width=0.75]    (10.93,-3.29) .. controls (6.95,-1.4) and (3.31,-0.3) .. (0,0) .. controls (3.31,0.3) and (6.95,1.4) .. (10.93,3.29)   ;
\draw  [dash pattern={on 4.5pt off 4.5pt}]  (640,190) -- (640,152) ;
\draw [shift={(640,150)}, rotate = 90] [color={rgb, 255:red, 0; green, 0; blue, 0 }  ][line width=0.75]    (10.93,-3.29) .. controls (6.95,-1.4) and (3.31,-0.3) .. (0,0) .. controls (3.31,0.3) and (6.95,1.4) .. (10.93,3.29)   ;
\draw  [dash pattern={on 4.5pt off 4.5pt}]  (100,190.4) -- (100,160) -- (100,152.4) ;
\draw [shift={(100,150.4)}, rotate = 90] [color={rgb, 255:red, 0; green, 0; blue, 0 }  ][line width=0.75]    (10.93,-3.29) .. controls (6.95,-1.4) and (3.31,-0.3) .. (0,0) .. controls (3.31,0.3) and (6.95,1.4) .. (10.93,3.29)   ;
\draw    (320,120) -- (320,140) ;
\draw   (380.6,111.85) .. controls (380.59,107.18) and (378.25,104.86) .. (373.58,104.87) -- (363.75,104.89) .. controls (357.08,104.91) and (353.74,102.59) .. (353.73,97.92) .. controls (353.74,102.59) and (350.42,104.93) .. (343.75,104.95)(346.75,104.94) -- (330.31,104.98) .. controls (325.64,104.99) and (323.32,107.33) .. (323.33,112) ;
\draw    (450,120) -- (450,140) ;
\draw    (530,120) -- (530,140) ;
\draw  [dash pattern={on 4.5pt off 4.5pt}]  (380,190) -- (380,152) ;
\draw [shift={(380,150)}, rotate = 90] [color={rgb, 255:red, 0; green, 0; blue, 0 }  ][line width=0.75]    (10.93,-3.29) .. controls (6.95,-1.4) and (3.31,-0.3) .. (0,0) .. controls (3.31,0.3) and (6.95,1.4) .. (10.93,3.29)   ;
\draw  [dash pattern={on 4.5pt off 4.5pt}]  (450,190) -- (450,152) ;
\draw [shift={(450,150)}, rotate = 90] [color={rgb, 255:red, 0; green, 0; blue, 0 }  ][line width=0.75]    (10.93,-3.29) .. controls (6.95,-1.4) and (3.31,-0.3) .. (0,0) .. controls (3.31,0.3) and (6.95,1.4) .. (10.93,3.29)   ;
\draw  [dash pattern={on 4.5pt off 4.5pt}]  (530,190) -- (530,152) ;
\draw [shift={(530,150)}, rotate = 90] [color={rgb, 255:red, 0; green, 0; blue, 0 }  ][line width=0.75]    (10.93,-3.29) .. controls (6.95,-1.4) and (3.31,-0.3) .. (0,0) .. controls (3.31,0.3) and (6.95,1.4) .. (10.93,3.29)   ;
\draw   (639.6,117.85) .. controls (639.56,113.18) and (637.21,110.87) .. (632.54,110.92) -- (595.03,111.26) .. controls (588.37,111.32) and (585.02,109.02) .. (584.97,104.35) .. controls (585.02,109.02) and (581.71,111.38) .. (575.04,111.45)(578.04,111.42) -- (537.54,111.79) .. controls (532.87,111.84) and (530.56,114.19) .. (530.6,118.86) ;

\draw (124,42.4) node [anchor=north west][inner sep=0.75pt]    {$\frac{1}{\alpha ^{\frac{m}{2}}} +1$};
\draw (552,52.4) node [anchor=north west][inner sep=0.75pt]    {$\frac{2}{\alpha ^{\frac{m}{2} -1}} -2$};
\draw (211,42.4) node [anchor=north west][inner sep=0.75pt]    {$\frac{1}{\alpha ^{\frac{m}{2}}} -1$};
\draw (331,52.4) node [anchor=north west][inner sep=0.75pt]    {$\frac{2}{\alpha } -2$};
\draw (191,190.4) node [anchor=north west][inner sep=0.75pt]    {$i_{1}$};
\draw (260,190.4) node [anchor=north west][inner sep=0.75pt]    {$i_{2}$};
\draw (311,190.4) node [anchor=north west][inner sep=0.75pt]    {$b_{1}$};
\draw (91,190.4) node [anchor=north west][inner sep=0.75pt]    {$a_{j's}$};
\draw (51,22.4) node [anchor=north west][inner sep=0.75pt]    {$d:$};
\draw (627,190.4) node [anchor=north west][inner sep=0.75pt]    {$b_{\frac{m}{2}}$};
\draw (291,72.4) node [anchor=north west][inner sep=0.75pt]    {$2$};
\draw (391,52.4) node [anchor=north west][inner sep=0.75pt]    {$\frac{2}{\alpha ^{2}} -2$};
\draw (442,190.4) node [anchor=north west][inner sep=0.75pt]    {$b_{3}$};
\draw (371,190.4) node [anchor=north west][inner sep=0.75pt]    {$b_{2}$};
\draw (511,187.4) node [anchor=north west][inner sep=0.75pt]    {$b_{\frac{m}{2} -1}$};
\draw (473,92.4) node [anchor=north west][inner sep=0.75pt]    {$.....$};
\draw (473,182.4) node [anchor=north west][inner sep=0.75pt]    {$.....$};

\end{tikzpicture}
}
\vspace{-10mm}
\caption{Metric used in the proof of Theorem \ref{general_lower_bound for line}.}
\label{fig:metric d}
\end{figure}

 To give a better sense of how this bound looks, we provided and illustration of this bound for different values of $m$ and $\alpha$ in Figures~\ref{illustrations-line:up1} and \ref{illustrations-line:up2}.




\begin{figure}[tb]
\centering

\begin{minipage}{0.48\textwidth}
  \centering
  \includegraphics[width=0.95\textwidth]{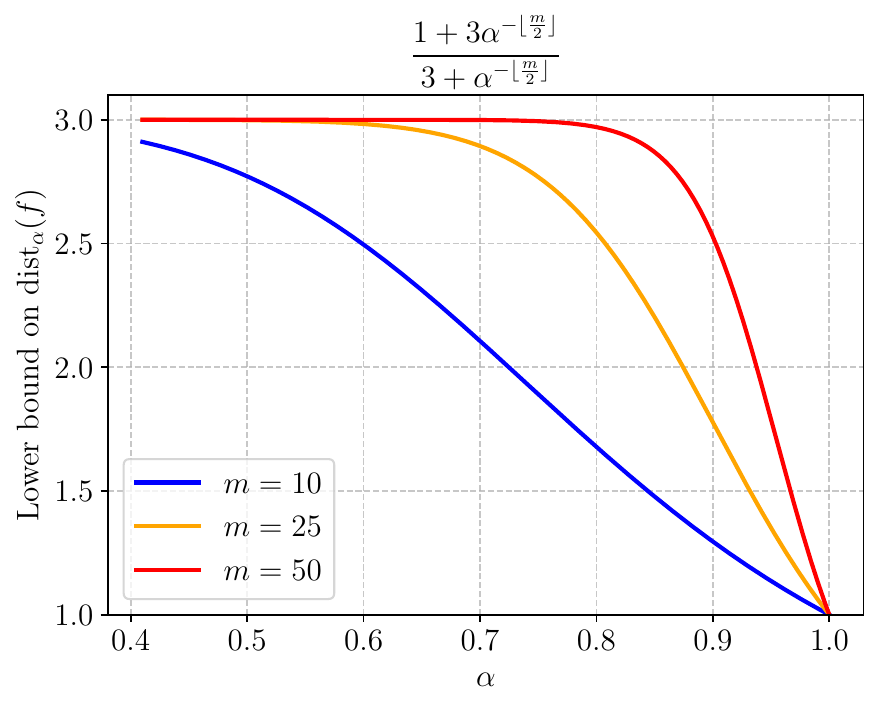}
   \caption{An illustration of the lower bounds established in \Cref{thm:metric_ub} for $10$, $25$, $50$ alternatives.}
   \label{illustrations-line:up1}
\end{minipage}
\hfill
\begin{minipage}{0.48\textwidth}
  \centering
  \includegraphics[width=0.95\textwidth]{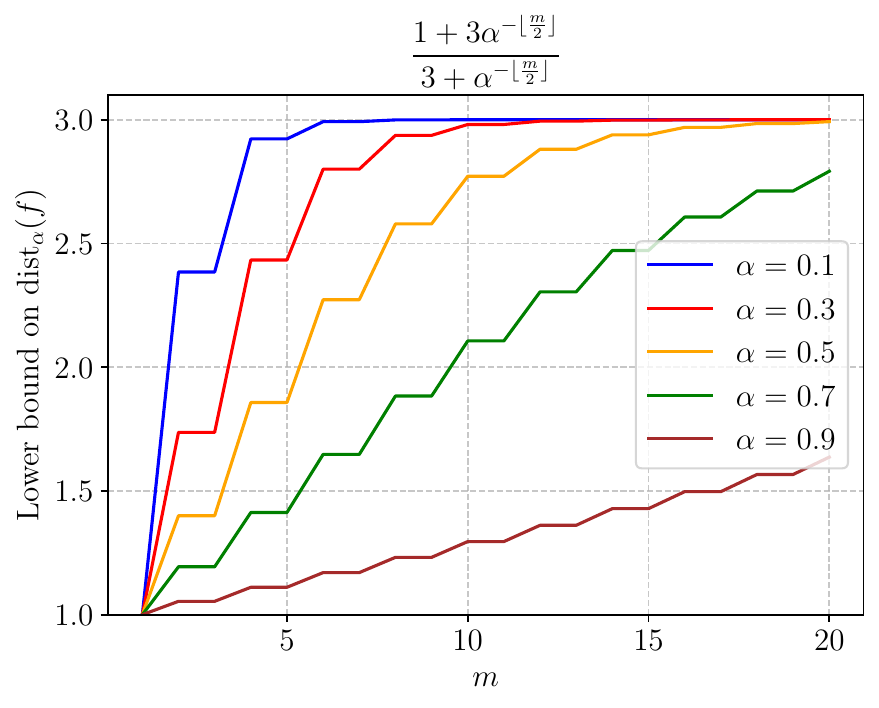}
   \caption{An illustration of the lower bounds established in \Cref{thm:metric_ub} based on the number of alternatives.}
   \label{illustrations-line:up2}
\end{minipage}

\end{figure}

\section{Price of Ignoring Intensities} \label{sec:poi}

In this section, we explore the concept of the \textit{Price of Ignoring Intensities (POII)}, which is a measure for the efficiency loss that rules that do not elicit the intensities incur. This notion was initially introduced by \citet{KLS23}, and is defined as the worst case over all elections, the ratio of the optimal distortion achievable without any information on the intensities to the optimal distortion achievable knowing the intensities. Roughly speaking, the \textit{Price of Ignoring Intensities} measures the increase in distortion when the intensities are not known.

To formally define POII, we have to first distinguish between the optimal alternative with and without information on the intensities. We borrow these definitions from \citet{KLS23}.

\begin{definition}[Intensity-aware optimal]
    For election $\elc = (\ags, \alts, \prefp)$, the intensity-aware optimal alternative, $\optaw(\elc)$, is an alternative that minimizes the worst-case distortion over the metrics that are $\alpha$-consistent with $\prefp$, i.e.,
    \vspace{-1mm}
	 $$\optaw(\elc) := \argmin_{a \in \alts} \dist_\alpha(a, \elc).$$
\end{definition}


Now, we define the POII of an alternative in an election.

\begin{definition}[Price of ignoring the intensities]
	For fixed $\alpha \in [0, 1]$, we define the price of ignoring intensities (POII) of an alternative $a \in \alts$ in election $\elc = (\ags, \alts, \prefp)$ as the ratio between the $\alpha$-distortion of $a$ and that of the intensity-aware optimal alternative, i.e.,
    \vspace{-2mm}
	$$
	\poi(a,\elc,\alpha) := \frac{\dist_{\alpha}(a,\elc)}{\dist_{\alpha}(\optaw(\elc),\elc)}\cdot
	$$
	In addition, given only the ranked preference profile (without intensities) $\vec{\pi}$, POII of alternative $a$ on $\vec{\pi}$ is defined as
	$
	\poi(a,\vec{\pi},\alpha) = \max_{\elc \rhd \vec \pi} \poi(a,\elc,\alpha)
	$, 
	  where for $\elc = (\ags, \alts, \prefp)$ we use $\elc \rhd \vec{\pi}$ to denote that $\prefp = (\vec{\pi},\vec{\intns})$ for some $\vec{\intns}$. This allows us to define the \emph{intensity-oblivious optimal} alternative in $\elc$ as $$\optob(\elc) := \argmin_{a \in \alts} \poi(a,\vec{\pi},\alpha),$$ and the POII on $\elc$ as $\poi(\elc,\alpha) := \poi(\optob(\elc),\elc,\alpha)$. Since $\optob(\elc)$ does not depend on $\vec \intns$, when other components are clear from the context, with a slight abuse of notion we use $\optob(\vec{\pi})$ instead. We are interested in the worst case of this over all $\elc$, termed the \emph{POII for $\alpha$-decisive preferences}: $\poi(\alpha) = \max_{\elc} \poi(\elc,\alpha)$. 
\end{definition}


Our goal is to prove a lower bound on the POII. To do so, we use the following two lemmas. 

\begin{lemma} \label{lb POI}
(\cite{KLS23}. Lemma 1) For $\alpha \in [0, 1]$, any election $\elc$, and alternative $a \in \alts$ we have:
$$\poi(\elc, \alpha) \geq  \frac{\dist_\alpha(\optob(\elc), \elc)}{\dist_\alpha(a, \elc)}\cdot$$
\end{lemma}

\begin{restatable}{lemma}{PolarExactUp}
\label{polar-exact-up}
Consider an election $\elc$ with $m$ alternatives, where $m$ is an even number, and two agents with the following preferences:
   $ \pref_1 = a_1  \,\ggc\, \cdots \,\ggc\, a_{\frac{m}{2}+1} \succ a_{\frac{m}{2} + 2} \succ \cdots \succ a_{m},$ and $
    \pref_2 = a_{\frac{m}{2}+1} \succ a_{\frac{m}{2}+2} \succ \cdots \succ a_{m} \succ a_{1} \succ \cdots \succ a_{\frac{m}{2}}\cdot
$
\noindent
We have
$\dist_\alpha(a_{1}, \elc) = \frac{1+2\alpha^{\frac{m}{2}}-\alpha^{m}}{\alpha^{m}+1}\cdot$ In addition,  when $m$ is an odd number, 
there exists a preference profile $\elc'$ such that $\dist_\alpha(a_{1}, \elc') = \frac{1+2\alpha^{ \frac{m-1}{2} }-\alpha^{ m -1 }}{\alpha^{m-1}+1}\cdot$\end{restatable}

\begin{proof}

Let us first give an overview of the proof. To maximize distortion, we have to choose $a_{\frac{m}{2}+1}$ as the best alternative. We then formulate a linear program to maximize the distortion, as introduced in \cite{kempe2020analysis}, with the following additional set of constraints corresponding to the intensities:

\begin{align*}
     d_{i, \pi_{i}(j)} &\leq \alpha d_{i, \pi_{i}(j + 1)} \quad & \forall i, j : \intns_i(j) = \qggc
    \\ 
     \alpha d_{i, \pi_{i}(j + 1)} &\leq d_{i, \pi_{i}(j)} \quad & \forall i, j : \intns_i(j) = \qsucc
\end{align*}

 Next, we write the dual of the LP and show that, for this instance, maximizing the objective in the primal LP requires that all variables be strictly greater than zero. This allows us to apply the conditions of Complementary Slackness to prove that all inequalities in the dual LP must be tight. Consequently, we can find a dual feasible solution equal to the distortion and thus prove the theorem.

The following analysis assumes that $m$ is even. We will address the case where $m$ is odd at the end. We want to determine a metric \( d \) in which \( a_{\frac{m}{2}+1} \) is the optimal alternative, and the ratio 
\[\frac{\scsd{d}(a_1)}{\scsd{d}(a_{\frac{m}{2} + 1})}\]
is maximized. It is important to note that the alternative causing the highest distortion relative to \( a_1 \) is \( a_{\frac{m}{2}+1} \), since for any other alternative $b$, either $d(1, b) \geq d(1, a_1)$ or $d(2, b)\geq d(2, a_{\frac{m}{2}+1})$. The LP formulation is as follows:
\begin{tcolorbox}[center,boxsep=1pt,top=2pt,bottom=5pt, width=1\textwidth]
\centerline{\underline{Linear Program \PP}}
\centering
\begin{align}
    \nonumber
    \text{\small Maximize} \quad  & d_{1, a_{1} } + d_{2, a_{1} }
    \\ 
    \nonumber
    \text{\small subject to} \quad & d_{i,a_j} \leq d_{i,a_{j'}} + d_{i',a_{j'}} + d_{i', a_j} \quad & \text{$1 \leq i \neq i' \leq 2$,\;  $1 \leq j \neq j' \leq m$ (\small triangle inequality)}
    \\ 
    \nonumber
    & d_{i, \pi_{i}(j)} \leq d_{i, \pi_{i}(j+1)} \quad & \text{$i \in \{1, 2\}$,\; $1 \leq j \leq m - 1$ (preference constraint)}
    \\ 
    \nonumber
    & d_{i, \pi_{i}(j)} \leq \alpha d_{i, \pi_{i}(j + 1)}  & \forall i, j : \intns_i(j) = \qggc \text{\small(intensity constraint)}
    \\ 
    \nonumber
    & \alpha d_{i, \pi_{i}(j + 1)} \leq d_{i, \pi_{i}(j)} & \forall i, j : \intns_i(j) = \qsucc \text{\small(intensity constraint)}
    \\ 
    \nonumber 
    &d_{1, a_{\frac{m}{2}+1}} + d_{2, a_{\frac{m}{2}+1}} = 1 \quad & \text{\small(normalization)}\\
    \nonumber 
    & d_{1, a_j} + d_{2, a_j} \geq 1 \quad & \text{$1 \leq j \leq m ,\, j \neq \frac{m}{2}+1$ (\small optimality of $a_{\frac{m}{2}+1}$)}
    \\
    & d_{i, a_j} \geq 0 \quad &\text{$i \in \{1, 2\}$,\; $1 \leq j \leq m - 1$. \nonumber}
\end{align}
\end{tcolorbox}
In LP \PP, variable $d_{i, a_j}$ denote the distances between agent $i$ and alternative $a_j$. These variables must be non-negative, satisfy the triangle inequality, and be compatible with the given election in the mandatory setting. Assuming the costs are normalized, we have $\scsd{d}(a_{\frac{m}{2}+1})=1$.

Now, let us take the dual of LP \PP. The dual includes four types of variables: $\psi^{i, i'}_{a_{j}, a_{j'}}$ are the dual variables referring to the triangle inequality constraints, $\phi^{i}_{\pi_{i}(j), \pi_{i}(j + 1)}$ are the dual variables for the preference constraints, $\beta^{i}_{\pi_{i}(j), \pi_{i}(j + 1)}$ are the dual variables for the intensity constraints, and $t_{a_j}$ are the dual variables for the normalization/optimality constraints. 

\begin{tcolorbox}[center,boxsep=2pt,top=2pt,bottom=5pt, width=0.8\textwidth]
\centerline{\underline{Linear Program \QQ}}
\centering
\begin{align}
    \nonumber 
    \text{Minimize} \quad & \sum_{j = 1}^{m} t_{a_j}
    \\ 
    \label{poi-dual-constraint}
    \text{subject to} \quad & \theta(i, j) \geq 
     \begin{cases}
         1 & \pi_i(j) = a_{1}\\
         0 & \pi_i(j) \neq a_{1}
     \end{cases} 
     & \text{ $i \in \{1, 2\}$,\; $1 \leq j \leq m$}
     \\ 
     \nonumber
     &\psi^{i, i'}_{a_{j}, a_{j'}} \geq 0 &  \text{$1 \leq i \neq i' \leq 2$,\; $1 \leq j \neq j' \leq m$}\\
     \nonumber
     & \phi^{i}_{\pi_{i}(j), \pi_{i}(j + 1)} \geq 0  & \text{$i \in \{1, 2\}$,\; $1 \leq j \leq m - 1$}\\ 
     \nonumber
     & \beta^{i}_{\pi_{i}(j), \pi_{i}(j + 1)} \geq 0  & \text{$i \in \{1, 2\}$,\; $1 \leq j \leq m - 1$}\\
     \nonumber
     & t_{a_j} \leq 0  & \text{$a_j \neq a_{\frac{m}{2}+1}$},
\end{align}
\end{tcolorbox}
where
\begin{align}
    \nonumber
    \theta(i, j) = \; &
     t_{\pi_{i}(j)} +  (\mathbf{1}_{\{j \neq m\}}) \phi^{i}_{\pi_{i}(j), \pi_{i}(j + 1)} - (\mathbf{1}_{\{j \neq 1\}}) \phi^{i}_{\pi_{i}(j - 1), \pi_{i}(j)} 
     \\   
    \nonumber
     &  + \sum_{\substack{1 \leq i' \leq 2, \; 1 \leq j' \leq  m \\ i'\neq i , \; \pi_{i'}(j') \neq \pi_{i}(j)}} \left(\psi^{i, i'}_{\pi_{i}(j), \pi_{i'}(j')} - \psi^{i, i'}_{\pi_{i'}(j'), \pi_{i}(j)}  - \psi^{i', i}_{\pi_{i}(j), \pi_{i'}(j')} - \psi^{i', i}_{\pi_{i'}(j'), \pi_{i}(j)} \right)
     \\ \nonumber 
     & +\begin{cases}
       (\mathbf{1}_{\{j \neq m\}}) \frac{1}{\alpha} \beta^{i}_{\pi_{i}(j), \pi_{i}(j + 1)} - (\mathbf{1}_{\{j \neq 1\}}) \beta^{i}_{\pi_{i}(j - 1), \pi_{i}(j)}  & i = 1, j \leq \frac{m}{2} \\
         - (\mathbf{1}_{\{j \neq 1\}}) \beta^{i}_{\pi_{i}(j - 1), \pi_{i}(j)} - (\mathbf{1}_{\{j \neq m\}}) \frac{1}{\alpha} \beta^{i}_{\pi_{i}(j), \pi_{i}(j + 1)}  & i = 1, j = \frac{m}{2} + 1\\
        (\mathbf{1}_{\{j \neq 1\}}) \beta^{i}_{\pi_{i}(j - 1), \pi_{i}(j)} - (\mathbf{1}_{\{j \neq m\}}) \frac{1}{\alpha} \beta^{i}_{\pi_{i}(j), \pi_{i}(j + 1)}  & \text{otherwise.}
    \end{cases}
\end{align}
We now show that to maximize the distortion, we must have $d(i, a_j) >0 \,\, \forall i, j$. First, note that if $d(2, b) = 0$ for some $b \in \alts$, then $d(2, b) = 0$ for all $b \in \alts$. To satisfy the triangle inequality in this case, $d(1, b)$ must be the same for all $b$, which results in a distortion of $1$. Hence, $d(2, b) > 0$ for all $b \in \alts$.
Now, suppose that $i$ is the smallest index for which $d(1, a_i) > 0$. Note that $1 \leq i \leq \frac{m}{2}+1$, because if not, we can deduce that the distortion is $1$ using the similar argument. Since none of the triangle inequalities with $d(1, a_j)$ in the left-hand side for any $1\leq j \leq i-1$ are tight (as we already know $d(2, b) > 0$), we can choose a small number $\epsilon$, and for all $1 \leq j \leq i-1$, increase $d(1, a_j)$ by $\epsilon (\frac{1}{\alpha})^{j-1}$ while still satisfying the triangle inequalities. One can easily check that this increase of values does not violate the intensity constraints. If $i > 1$, $d(1, a_1)$ increases, leading to an increase in the distortion of $a_1$. Hence, $i = 1$, and we can deduce that $d(i, a_j) >0 \,\, \forall i, j$. This fact, along with Complementary Slackness conditions, show that Constraint \eqref{poi-dual-constraint} must hold as an equality. We use this to introduce the feasible solution of LP \QQ.
\\ The non-zero variables of the feasible solution $(\psi, \phi, \beta, t)$ for LP \QQ are as follows: 
\begin{align}
    \nonumber
    &\beta^{2}_{\pi_{2}(\frac{m}{2}-i), \pi_{2}(\frac{m}{2}+1-i)} = y \left(\frac{1}{\alpha} \right)^i
    & 0\leq i \leq \frac{m}{2}-1 
    \\ 
    \nonumber
    &\beta^{1}_{\pi_{1}(\frac{m}{2}-i), \pi_{1}(\frac{m}{2}+1-i)} = y \left( \frac{1}{\alpha}\right)^{\frac{m}{2}+i}
    & 0 \leq i \leq \frac{m}{2}-1
    \\ 
    \nonumber
    &\psi^{i, i'}_{a_{j}, a_{j'}} = 1-y
    & i = 2, \; i' = 1, \;, a_j = a_1, \; a_{j'} =a_{\frac{m}{2}+1} 
    \\ 
    \nonumber
    &t_{a_j} = y  \left( \left(\frac{1}{\alpha}\right)^{\frac{m}{2}}-1\right) +1 
    & a_j = a_{\frac{m}{2}+1},
\end{align}
where $y = \frac{2}{ \left( \frac{1}{\alpha} \right)^m+1}$.
We can check that these variables satisfy constraint \eqref{poi-dual-constraint}.
Therefore, $(\psi, \phi, \beta, t)$ is a feasible solution for LP \QQ, as desired.\\
Now, we introduce a feasible solution for LP \PP. The variables of the feasible solution $d$ for LP \PP are as follows:
\begin{equation}
    \nonumber d_{i, \pi_i(j)} =
    \begin{cases}
         \frac{ \left(\frac{1}{\alpha}\right)^{\frac{m}{2}+j-1}+ \left(\frac{1}{\alpha}\right)^{j-1} }{\left(\frac{1}{\alpha}\right)^m+1 }, \; 
         &i=2,\; 1\leq j \leq \frac{m}{2}+1
         \\ 
         \frac{\left(\frac{1}{\alpha}\right)^m +\left(\frac{1}{\alpha}\right)^{\frac{m}{2}} }{\left(\frac{1}{\alpha}\right)^m+1 }, \; 
         &i=2,\; \frac{m}{2}+2\leq j \leq m
         \\ 
         \frac{ \left(\frac{1}{\alpha}\right)^{\frac{m}{2}+j-1}- \left(\frac{1}{\alpha}\right)^{j-1} }{\left(\frac{1}{\alpha}\right)^m+1 }, \; 
         &i=1,\; 1\leq j \leq \frac{m}{2}+1
         \\ 
         \frac{ \left(\frac{1}{\alpha}\right)^m- \left(\frac{1}{\alpha}\right)^{\frac{m}{2}} }{\left(\frac{1}{\alpha}\right)^m+1 }, \; 
         &i=1,\; \frac{m}{2}+2\leq j \leq m.
    \end{cases}
\end{equation}
Also, since all the variables are non-negative, we only need to check that they satisfy the triangle inequality. Specifically, for every $1 \leq i \neq i' \leq 2$ and $1 \leq j \neq j' \leq m$, we must have:
\begin{equation}
\label{poi-tri-ine-lp}
d_{i,a_j} \leq d_{i,a_{j'}} + d_{i',a_{j'}} + d_{i', a_j}.    
\end{equation}
We consider Inequality \eqref{poi-tri-ine-lp} in three cases. First, assume that $1 \leq j \neq j' \leq \frac{m}{2}$. In this case, at least one of the $d_{2, a_{j}}$ or $d_{2, a_{j'}}$ appears on the right-hand side if the inequality. Since both of them are the largest possible value that $d$ takes, inequality holds trivially. Second, assume that $\frac{m}{2}+1 \leq j \neq j' \leq m$. If $d_{1, a_{j}}$ appears on the left-hand side of the inequality, inequality holds trivially, since $d(1, a_{j})$ and $d(1, a_{j'})$ have the same value and the $d(1, a_{j'})$ appears on the other side of the inequality. If both of them appear on the right-hand side of the inequality, the worst case occurs when $a_j = \pi_{2}(\frac{m}{2})$ and $a_{j'} = \pi_{2}(1)$. We have
\begin{align}
    \nonumber
    d(2, \pi_{2}(\frac{m}{2})) &=  \frac{\left(\frac{1}{\alpha}\right)^m +\left(\frac{1}{\alpha}\right)^{\frac{m}{2}} }{\left(\frac{1}{\alpha}\right)^m+1 }\\
    \nonumber
    &\leq \frac{2\left(\frac{1}{\alpha}\right)^m - \left(\frac{1}{\alpha}\right)^\frac{m}{2} + 1}{\left(\frac{1}{\alpha}\right)^m+1 }    & \left(\text{since }\left(\left(\frac{1}{\alpha}\right)^\frac{m}{2} - 1\right)^2 \geq 0\right) \\
    \nonumber
    &= d(2, \pi_{2}(1)) + d(1, \pi_{2}(1)) + d(1, \pi_{2}(\frac{m}{2})).
\end{align}
Third, assume that either $1 \leq j \leq \frac{m}{2}$ and $\frac{m}{2}+1 \leq j' \leq m$, or $1 \leq j' \leq \frac{m}{2}$ and $\frac{m}{2}+1 \leq j \leq m$. In this case, the worst case occurs when $a_j = \pi_{1}(1)$ and $a_{j'} = \pi_{2}(1)$. Now, we have
\[
d(2, \pi_{1}(1)) = \frac{\left(\frac{1}{\alpha}\right)^m +\left(\frac{1}{\alpha}\right)^{\frac{m}{2}} }{\left(\frac{1}{\alpha}\right)^m+1 } = d(1, \pi_{2}(1)) + d(2, \pi_{2}(1)) + d(1, \pi_{1}(1)).
\]
Therefore, $d$ satisfies the triangle inequality and is a feasible solution for LP \PP.
\\Now, note that the objective value of LP \QQ, given by $(\psi, \phi, \beta, t)$, is: 
\[
\sum_{j = 1}^{m} t_{a_j} = t_{a_{\frac{m}{2}+1}} =\frac{2}{ \left( \frac{1}{\alpha} \right)^m+1}  \left( \left(\frac{1}{\alpha}\right)^{\frac{m}{2}}-1\right) +1 = 1 + \frac{2\alpha^{\frac{m}{2}}-2\alpha^{m}}{\alpha^{m}+1}.
\]
Similarly, the objective value of LP \PP, given by $d$, is:
\[
d_{2, a_{1}} + d_{1, a_{1}} = \frac{\left( \frac{1}{\alpha} \right)^m + \left( \frac{1}{\alpha} \right)^{\frac{m}{2}}}{\left( \frac{1}{\alpha} \right)^m+1} + \frac{\left( \frac{1}{\alpha} \right)^{\frac{m}{2}} -1}{\left( \frac{1}{\alpha} \right)^m+1}=1 + \frac{2\alpha^{\frac{m}{2}}-2\alpha^{m}}{\alpha^{m}+1} .
\]
Therefore, $1 + \frac{2\alpha^{\frac{m}{2}}-2\alpha^{m}}{\alpha^{m}+1} = \frac{1+2\alpha^{\frac{m}{2}}-\alpha^{m}}{\alpha^{m}+1} = \dist_\alpha(a_{1}, \elc)$ is the optimal value of both LP \QQ and LP \PP.
\\
Now, assume that $m$ is odd. Consider the election \(\elc' = (\ags, \alts, \prefp)\) where \(\ags = \{1, 2\}\) and \(\alts = \{a_1, a_2, \ldots, a_{m}\}\). The preferences of agents are as follows:
\begin{equation}
\nonumber \pref_{l}  =
\begin{cases}
        a_1 \;\ggc\; a_2 \;\ggc\; \ldots \;\ggc\; a_{\frac{m-1}{2}} \;\ggc\; a_{\frac{m-1}{2}+1} \succ a_{\frac{m-1}{2}+2} \succ \ldots \succ a_{m} &\; l = 1\\
        a_{\frac{m-1}{2}+1} \succ a_{\frac{m-1}{2}+2} \succ \ldots \succ a_{m-1} \succ a_1 \succ a_2 \succ \ldots \succ a_{\frac{m-1}{2}} \;\succ \; a_{m} &\; l = 2.
\end{cases}
\end{equation}
We drop $a_m$ from the election since $a_m$ is the least preferred alternative for each agent. Now, we can apply the metric $d$ we found for the case where $m$ is even to the remaining alternatives and agents, and obtain the distortion of $a_1$:
\[
\dist_\alpha(a_{1}, \elc') = \frac{1+2\alpha^{\lfloor \frac{m}{2} \rfloor}-\alpha^{\lfloor m \rfloor_{\text{even}}}}{\alpha^{\lfloor m \rfloor_{\text{even}}}+1}.
\]
To complete this argument, we have to determine  $d(1, a_m)$ and $d(2, a_m)$. If we set $d(1, a_m) = d(1, \pi_1(m-1))$ and $d(2, a_m) = d(2, \pi_2(m-1))$, it is easy to check that intensity and preference constraints are still satisfied.

\end{proof}

\begin{restatable}{theorem}{poiMandLb}
\label{thm:poi_mand_lb}  
For every $\alpha \in [0,1]$, we have
$$
\poi(\alpha) \geq
\frac{3 \left(\alpha^{\lfloor m \rfloor_{\text{even}}}+1\right)}{1+2\alpha^{\lfloor \frac{m}{2} \rfloor}-\alpha^{\lfloor m \rfloor_{\text{even}}}},
$$
 where $\lfloor m \rfloor_{\text{even}}$ is the largest even integer smaller than or equal to $m$.
\end{restatable}

\begin{proof}
    First, assume that $m$ is even. Consider an election $\elc$ with two agents and $m$ alternatives, where $\pi_1 = \left(a_1, \dots , a_m\right)$ and $\pi_2 = \left(a_{\frac{m}{2}+1}, \dots , a_m, a_1 , a_2, \dots, a_{\frac{m}{2}}\right)$. We determine intensities based on the position of $\optob(\vec \pi)$ in the rankings.
    If $\optob(\vec \pi)$ appears in the first half of the first agent's ranking, we set
    $a_1 \succ a_2 \succ \ldots \succ a_{m}$ for the first agent and $a_{\frac{m}{2}+1} \,\ggc \,a_{\frac{m}{2}+2} \,\ggc\, \ldots \,\ggc\, a_{m} \,\ggc\, a_1 \succ a_2 \succ \ldots \succ a_{m}$ for the second one. In the other case, we set
    $a_1 \,\ggc\, a_2 \,\ggc\, \ldots \,\ggc\, a_{\frac{m}{2}} \,\ggc\, a_{\frac{m}{2}+1} \succ a_{\frac{m}{2} + 2} \ldots \succ a_{m}$ for the first agent and $a_{\frac{m}{2}+1} \succ a_{\frac{m}{2}+2} \succ \ldots \succ a_{m} \succ a_1 \succ a_2 \succ \dots \succ a_{\frac{m}{2}}$ for the second one.\\ 
    Without loss of generality, assume that $\optob(\vec \pi)$ is in the second half of $\pi_1$. Using Lemma, \ref{lb POI} we have:
    \begin{equation}
    \label{eq:poii-app}
    \poi(\alpha) \geq  \frac{\dist_\alpha(\optob(\elc), \elc)}{\dist_\alpha(a_{1}, \elc)}\cdot
    \end{equation}
 
    To prove the desired bound, we calculate the exact amount of the right-hand side of Inequality 
    \eqref{eq:poii-app}. Consider the following metric:
    $$ d(1, b) = \begin{cases}
        0  & \rank{1}(b) \leq \frac{m}{2} \\ 2  & \rank{1}(b) > \frac{m}{2}
    \end{cases}, \text{ and } d(2, b) = 1 \cdot$$
It is easy to check that $d$ satisfies the triangle inequality and intensity constraints. Now, note that
\begin{align}
    \nonumber
    \dist_\alpha(\optob(\elc), \elc) &\geq \frac{\scs(\optob(\elc))}{\scs(a_1)} \\
    \nonumber
    &= \frac{d(1, \optob(\elc)) + d(2, \optob(\elc))}{d(1, a_1) + d(2, a_1)} = 3 & \left(\rank{1}(\optob(\elc)) > \frac{m}{2}\right). 
\end{align}
Hence, $\dist_\alpha(\optob(\elc), \elc)$ is at least $3$. But we know that it is also upper bounded by $3$. Therefore, it is exactly equal to $3$.
Note that Lemma \ref{polar-exact-up} gives the exact value of $\dist_\alpha(a_1, \elc)$, i.e., $\frac{1+2\alpha^{\frac{m}{2}}-\alpha^{m}}{\alpha^{m}+1}$. Now, using Inequality (\ref{eq:poii-app}) gives the desired lower bound on $\poi(\alpha)$.\\
In the case where $m$ is odd, we can use an election $\elc$ with two agents and $m$ alternatives, where, $\pi_1 = \left(a_1, \dots , a_{m-1}, a_{m}\right)$ and $\pi_2 = \left(a_{\frac{m-1}{2}+1}, \dots , a_{m-1}, a_1 , a_2, \dots, a_{\frac{m-1}{2}}, a_{m}\right)$. Now, we can drop $a_m$ from the election (since $a_m$  is the least preferred alternative for both agents) and apply a similar argument to the rest of the agents and alternatives to achieve the desired lower bound.  
\end{proof}


\subsection{Voluntary Elicitation of the Intensities}
\label{sec:vol}

As described in the preliminaries, we can make an alternate assumption on how the agents report their intensities. In this setting, we give the option of using $\qggc$ to the agents but they might decide not to use it. This can happen when agents do not want to report the intensity or are not sure if there is an $\alpha$-gap between the distances. Recall the definition of $\alpha$-consistency of a preference profile $\prefp$ and a metric space $d$ under voluntary assumption ($\prefp \cns d$). Using that definition, we can redefine the distortion for the voluntary setting:

$$\distv(a, \elc) := \sup_{d: \prefp \cns d} \frac{\scs_d(a)}{\min_{b \in \alts} \scs_d(b)}\cdot$$

Note that the optimal deterministic distortion in this setting is still 3, since if all the agents decide not to report their intensities, the voting rule is restricted to the data that is available in the classic voting with ranked preferences. However, we can use this definition to  define $\optawv, \optobv, \allowbreak \poiv(a, \elc, \alpha)$, and $\poiv(\alpha)$ similar to the mandatory setting, and then prove bounds on the POII under this assumption. Bounds on the POII in this setting show that even if the agents are not forced to report the intensities, if the voting rule does not take them into account, we might lose a great deal in terms of efficiency.

\begin{theorem}
For every $\alpha \in [0,1]$, we have:
$$
\poiv(\alpha) \geq \frac{3}{2 \alpha^{\lfloor\frac{m}{2}\rfloor} +1}\cdot
$$    
\end{theorem}

To prove this bound, we use the same instance as in the proof of \Cref{thm:poi_mand_lb}. However, the bound differs since the distortion of intensity aware optimal outputs in these two settings are different. The source of this difference is that the class of instances over which $\optaw$ has the optimal distortion is a subset of the instances that $\optawv$ has to be optimized for, since $\prefp \scns d$ is a more restricting condition on $d$ than $\prefp \cns d$.
Similar to the proof of \Cref{thm:poi_mand_lb}, we use an LP to formulate this setting and write the dual LP as well. However, in this case, we cannot guarantee that the variables are always positive and hence the Complementary Slackness conditions do not help. In the following proof, we present a feasible solution for the dual and we show an upper bound on the denominator.

\begin{proof} 
    We only consider the case where $m$ is even. Similar to the proof of Theorem \ref{thm:poi_mand_lb}, the odd case can be concluded from the even case.\\
    Consider an election $\elc$ with two agents and $m$ alternatives, where $\pi_1 = \left(a_1, \dots , a_m\right)$ and $\pi_2 = (a_{\frac{m}{2}+1}, \dots , a_m,\allowbreak  a_1 , a_2, \dots, a_{\frac{m}{2}})$. We determine intensities based on the position of $\optob(\vec{\pi})$ in the rankings.
        If $\optob(\vec{\pi})$ appears in the first half of the first agent's ranking, we set
    $a_1 \succ a_2 \succ \ldots \succ a_{m}$ for the first agent and $a_{\frac{m}{2}+1} \,\ggc \,a_{\frac{m}{2}+2} \,\ggc\, \ldots \,\ggc\, a_{m} \,\ggc\, a_1 \succ a_2 \succ \ldots \succ a_{m}$ for the second one. In the other case, we set
    $a_1 \,\ggc\, a_2 \,\ggc\, \ldots \,\ggc\, a_{\frac{m}{2}} \,\ggc\, a_{\frac{m}{2}+1} \succ a_{\frac{m}{2} + 2} \ldots \succ a_{m}$ for the first agent and $a_{\frac{m}{2}+1} \succ a_{\frac{m}{2}+2} \succ \ldots \succ a_{m} \succ a_1 \succ a_2 \succ \dots \succ a_{\frac{m}{2}}$ for the second one.\\ 
    Without loss of generality, assume that $\optob(\vec{\pi})$ is in the second half of $\pi_1$. Using Lemma, \ref{lb POI} we have:
    \begin{equation}
    \label{eq:poii-app-vol}
    \poi(\alpha) \geq  \frac{\dist_\alpha(\optob(\elc), \elc)}{\dist_\alpha(a_{1}, \elc)}\cdot
    \end{equation}\\
    Note that the metric $d$ provided in the proof of Theorem \ref{thm:poi_mand_lb} is also consistent with the intensity constraints in the voluntary setting. Hence, we can conclude that $\dist_\alpha(\optob(\elc), \elc) = 3$.\\
    Henceforth, we aim to find an upper bound on $\dist_\alpha(a_{1}, \elc)$. The technique we use is similar to the technique of the proof of Lemma \ref{polar-exact-up}. We write a linear program to find the metric $d$ that maximizes $\frac{\scs_d(a_1)}{\scs_d(a_{\frac{m}{2}+1})}$.

    \begin{tcolorbox}[center,boxsep=2pt,top=2pt,bottom=5pt, width=\textwidth]
\centerline{\underline{Linear Program \PP}}
\centering
\begingroup
\small
\begin{align}
    \nonumber
    \text{Maximize} \quad  & d_{1, a_{1} } + d_{2, a_{1} }
    \\ 
    \nonumber
    \text{subject to} \quad & d_{i,a_j} \leq d_{i,a_{j'}} + d_{i',a_{j'}} + d_{i', a_j} \quad & \text{$1 \leq i \neq i' \leq 2$,\;  $1 \leq j \neq j' \leq m$ (triangle inequality)}
    \\ 
    \nonumber
    & d_{i, \pi_{i}(j)} \leq \alpha d_{i, \pi_{i}(j + 1)}  & \forall i, j : \intns_i(j) = \qggc \text{(intensity constraint)}
    \\ 
    \nonumber
    & d_{i, \pi_{i}(j)} \leq d_{i, \pi_{i}(j+1)} & \forall i, j : \intns_i(j) = \qsucc \text{(intensity constraint)}
    \\ 
    \nonumber 
    &d_{1, a_{\frac{m}{2}+1}} + d_{2, a_{\frac{m}{2}+1}} = 1 \quad & \text{(normalization)}\\
    \nonumber 
    & d_{1, a_j} + d_{2, a_j} \geq 1 \quad & \text{$1 \leq j \leq m ,\, j \neq \frac{m}{2}+1$ (optimality of $a_{\frac{m}{2}+1}$)}
    \\
    & d_{i, a_j} \geq 0 \quad &\text{$i \in \{1, 2\}$,\; $1 \leq j \leq m - 1$. \nonumber}
\end{align}
\endgroup
\end{tcolorbox}
Note that in the voluntary setting, when we have $\intns_i(j) = \qsucc$, the only inference we can make is that $d_{i, \pi_{i}(j)} \leq d_{i, \pi_{i}(j+1)}$.\\
The dual includes three types of variables: $\psi^{i, i'}_{a_{j}, a_{j'}}$ are the dual variables referring to the triangle inequality constraints, $\phi^{i}_{\pi_{i}(j), \pi_{i}(j + 1)}$ are the dual variables for the intensity constraints, and $t_{a_j}$ are the dual variables for the normalization/optimality constraints.
\begin{figure*}
\begin{tcolorbox}[center,boxsep=2pt,top=2pt,bottom=5pt, width=0.8\textwidth]
\centerline{\underline{Linear Program \QQ}}
\centering
\begin{align*}
    \nonumber 
    \text{Minimize} \quad & \sum_{j = 1}^{m} t_{a_j}
    \\ 
    \label{poi-dual-constraint-vol}
    \text{subject to} \quad & \theta(i, j) \geq 
     \begin{cases}
         1 & \pi_i(j) = a_{1}\\
         0 & \pi_i(j) \neq a_{1}
     \end{cases} 
     & \text{ $i \in \{1, 2\}$,\; $1 \leq j \leq m$}
     \\ 
     \nonumber
     &\psi^{i, i'}_{a_{j}, a_{j'}} \geq 0 &  \text{$1 \leq i \neq i' \leq 2$,\; $1 \leq j \neq j' \leq m$}\\
     \nonumber
     & \phi^{i}_{\pi_{i}(j), \pi_{i}(j + 1)} \geq 0  & \text{$i \in \{1, 2\}$,\; $1 \leq j \leq m - 1$}\\ 
     \nonumber
     & t_{a_j} \leq 0  & \text{$a_j \neq a_{\frac{m}{2}+1}$},
\end{align*}
\end{tcolorbox}
\end{figure*}
where
\begin{align}
    \nonumber
    \theta(i, j) = \; &
     t_{\pi_{i}(j)}   + \sum_{\substack{1 \leq i' \leq 2, \; 1 \leq j' \leq  m \\ i'\neq i , \; \pi_{i'}(j') \neq \pi_{i}(j)}} \left(\psi^{i, i'}_{\pi_{i}(j), \pi_{i'}(j')} - \psi^{i, i'}_{\pi_{i'}(j'), \pi_{i}(j)}  - \psi^{i', i}_{\pi_{i}(j), \pi_{i'}(j')} - \psi^{i', i}_{\pi_{i'}(j'), \pi_{i}(j)} \right)
     \\ \nonumber 
     & +\begin{cases}
       (\mathbf{1}_{\{j \neq m\}}) \frac{1}{\alpha} \phi^{i}_{\pi_{i}(j), \pi_{i}(j + 1)} - (\mathbf{1}_{\{j \neq 1\}}) \phi^{i}_{\pi_{i}(j - 1), \pi_{i}(j)}  &  i = 1, j \leq \frac{m}{2} \\
       (\mathbf{1}_{\{j \neq m\}}) \phi^{i}_{\pi_{i}(j), \pi_{i}(j + 1)} - (\mathbf{1}_{\{j \neq 1\}}) \phi^{i}_{\pi_{i}(j - 1), \pi_{i}(j)}  &
       \text{otherwise}
    \end{cases}
\end{align}
Now, we introduce a feasible solution for LP \QQ. The non-zero variables of this feasible solution are as follows:
\begin{align}
    \nonumber
    &\phi^{1}_{\pi_{2}(\frac{m}{2}-i), \pi_{2}(\frac{m}{2}+1-i)} = x \left(\frac{1}{\alpha} \right)^i
    & 0\leq i \leq \frac{m}{2}-1 
    \\ 
    \nonumber
    &\psi^{i, i'}_{a_{j}, a_{j'}} = 1
    & i = 2, \; i' = 1, \;, a_j = a_1, \; a_{j'} =a_{\frac{m}{2}+1} 
    \\ 
    \nonumber
    &t_{a_j} = 1 + x 
    & a_j = a_{\frac{m}{2}+1},
\end{align}
where $x = 2 \alpha ^ {\frac{m}{2}}$.
The objective value of LP \QQ, given by this feasible solution, is:
\[
\sum_{j=1}^{m} t_{a_j} = t_{a_{\frac{m}{2}+1}} = 1 + 2 \alpha ^ {\frac{m}{2}}.
\]
Therefore, $\dist_{\alpha}(a_1, \elc)$ is upper bounded by $1 + 2 \alpha ^ {\frac{m}{2}}$. As a result, we have

\[
\poiv(\alpha) \geq \frac{3}{2 \alpha^{\frac{m}{2}} +1}\cdot
\]

\end{proof}
\section{Discussion and Future Work}
In this paper, we investigate the effect of adding intensities to voting with ranked ballots to the efficiency from the viewpoint of metric distortion. We established upper and lower bounds on the distortion when agents are required to report the intensities. With these results and the results of \citet{KLS23}, the next step would be to see if the best of both worlds (utilitarian and metric) guarantees can be made using the ideas of \citet{gkatzelis2023best}.

In this paper, our focus was solely on deterministic rules. While deterministic rules are more desirable in elections with high stakes, randomized rules have gained a lot of attention in repeated elections where the price of having one bad output once in a while is not that high. An immediate extension to our work is to see how randomization helps improve the distortion in this model.

One can also think about an axiomatic approach to voting with intensities. This ballot format provides extra information compared to both approvals and ranked preferences. An interesting direction would be to find a meaningful way to define axioms for this type of voting and see if desirable axioms are satisfiable.

Another direction is to bring this ballot format to other problems. Distortion has been defined and extensively analyzed in other settings such as fair division and matching. An interesting future work would be to see how intensities change the distortion in these problems.

\clearpage

\section*{Acknowledgements}
This work was completed while Mehrad Abbaszadeh was a student at Sharif University of Technology, Iran. Mohamad Latifian was supported by the UK Engineering and Physical Sciences Research Council (EPSRC) grant EP/Y003624/1.


\clearpage



\bibliographystyle{ACM-Reference-Format} 
\bibliography{abb,bib}

\clearpage



\end{document}